\documentclass[onefignum,onetabnum]{siamart190516}

\usepackage[utf8]{inputenc}
\usepackage[english]{babel}

\usepackage{amsmath, amssymb,mathtools}
\usepackage{cancel}
\usepackage{xfrac}		
\usepackage{mathdots}	
\usepackage{thmtools} 	
\usepackage{tcolorbox}	

\makeatletter
\renewcommand*\env@matrix[1][\arraystretch]{%
  \edef\arraystretch{#1}%
  \hskip -\arraycolsep
  \let\@ifnextchar\new@ifnextchar
  \array{*\c@MaxMatrixCols c}}
\makeatother

\usepackage{graphicx}
\usepackage{subcaption}
\usepackage{pgfplots}	

\usepackage{color}

\usepackage{tikz}
\usetikzlibrary{calc}
\usetikzlibrary{decorations.markings}
\usetikzlibrary{shapes.geometric,positioning}
\usetikzlibrary{arrows}
\tikzset{vertex/.style={minimum size=2mm,circle,fill=black,draw,inner sep=0pt},
         decoration={markings,mark=at position .5 with {\arrow[black,thick]{stealth}}}}
\pgfplotsset{compat=1.10}
\usepgfplotslibrary{fillbetween}
\usetikzlibrary{backgrounds}
\usetikzlibrary{patterns}

\newcommand{\R}{ \mathbb{R} }
\newcommand{\N}{ \mathbb{N} }
\newcommand{\C}{ \mathbb{C} }

\newtheorem{remark}[theorem]{Remark}
\definecolor{myred}{rgb}{0.86,0.14,0.14}
\definecolor{mygreen}{rgb}{.3,.9,.4}
\definecolor{darkmagenta}{rgb}{0.55, 0.0, 0.55}

\definecolor{julien}{rgb}{1, 0, 0}

\definecolor{lucas}{rgb}{1, 0, 1}

\title{{Parallel approximation of the exponential of
Hermitian matrices}} 
\author{
Frédéric Hecht\thanks{Laboratoire Jacques-Louis Lions, Sorbonne Universit\'e, CNRS,  75005 Paris and INRIA  Paris,  ALPINES  Project-Team,  75589  Paris  Cedex  12,  France(\email{frederic.hecht@sorbonne-universite.fr})}.
\and Sidi-Mahmoud Kaber\thanks{Laboratoire Jacques-Louis Lions, Sorbonne Universit\'e, CNRS,  75005 Paris, France (\email{sidi-mahmoud.kaber@sorbonne-universite.fr})}.
\and Lucas Perrin\thanks{INRIA  Paris,  ANGE  Project-Team,  75589  Paris  Cedex  12,  France
   and Sorbonne Universit\'e, CNRS, Laboratoire Jacques-Louis Lions, 75005 Paris, France (\email{lucas.perrin@inria.fr})}.
\and Alain Plagne\thanks{Centre de Mathématiques Laurent Schwartz, \'Ecole polytechnique, F-91128 Palaiseau, France. (\email{alain.plagne@polytechnique.edu})}.
\and Julien Salomon\thanks{INRIA  Paris,  ANGE  Project-Team,  75589  Paris  Cedex  12,  France
   and Laboratoire Jacques-Louis Lions, Sorbonne Universit\'e, CNRS, 75005 Paris, France (\email{julien.salomon@inria.fr})}.
}

\begin{document}

\maketitle

\begin{abstract}
In this work, we consider a rational approximation of the exponential function to design an algorithm for computing matrix exponential in the Hermitian case. Using partial fraction decomposition, we obtain  a parallelizable method, where the computation reduces to independent resolutions of linear systems. We analyze the effects of rounding errors on the accuracy of our algorithm. We complete this work with numerical tests showing the efficiency of our method and a comparison of its performances with Krylov algorithms.
\end{abstract}

\begin{keywords}
{Matrix exponential, Parallel computing, Truncation error, Taylor series, Partial fraction decomposition, Pad\'e approximation, MATLAB, Octave,  \texttt{expm}, Roundoff error.}

\end{keywords}


\section{Introduction}
%

Given a square matrix $A$, the differential equation $u'(t)=Au(t)$ appears in many models, either directly or  
 as an elementary component of more complicated differential systems. 
To solve this equation with a good accuracy, it is useful to have an algorithm computing matrix exponential.
This algorithm must 
be efficient, both for the accuracy and for the computational efficiency. Such an algorithm is presented in this paper.

Many algorithms for computing the exponential 
of a matrix are available. We refer to the celebrated review by Moler and Van Loan~\cite{comp} for a comparison of these methods. 
None of them is clearly more efficient than the others if we take into account various important criteria such as accuracy, computing time, memory space requirements, complexity, properties of the matrices under consideration, etc. 

As is the case with our method, several algorithms are based on rational approximation of the exponential function $e^z$ ($z\in{\mathbb C}$), such as Pad\'e or uniform Chebyshev approximations.
 Let $R_{m,n}(z)$ denotes such an approximation  ($m$ and $n$ are the  degrees of the numerator $N_{m,n}$ and denominator $D_{m,n}$ respectively), the considered approximation of $\exp(A)$ is given by $R_{m,n}(A)=[D_{m,n}(A)]^{-1}N_{m,n}(A)$. 

In the literature, such approximations are combined with scaling or reduction techniques, which mainly consider the so called \emph{diagonal case}, i.e., $m=n$. In~\cite{high}, the authors 
use the scaling and squaring method~\cite{highS} to compute $\exp(A) \simeq [R_{n,n}(A/\ell)]^\ell$ where  $R_{n,n}(z)$ is accurate enough near the origin to guarantee high order approximation of $\exp(A/\ell)$ with $\ell\in{\mathbb N}$. This strategy avoids the conditioning problem of $D_{n,n}$ that can occur when $n$ is large~\cite{comp}.
Rational approximation can also be applied to, e.g.,  reduced or 
simpler forms of $A$, as in~\cite{lu} where an orthogonal  
factorization $A=Q^TTQ$ is used to compute $\exp(A)\simeq Q^TR_{n,n}(T)Q$ with $T$ a tridiagonal matrix.

In the case $n=0$, properties of orthogonal polynomials have been used to define approximations of the matrix exponential, see~\cite{MR1751993} for the Chebyshev case (orthogonality on a bounded interval) and~\cite{MR2777241} for the Laguerre case (orthogonality on the half real line). The interest, in both cases, relies on saving in storage of Ritz vectors during the Krylov iterations.


Our algorithm is based on an independent approach which aims at decomposing the computation in view of parallelization. As a consequence, it can be combined with all the previous algorithms which require the computation of the exponential of a transformed matrix. 
It shares some features  with the one presented in~\cite{saadP} and~\cite{saad}, where diagonal approximations are used to define implicit numerical schemes for linear parabolic equations. The parallelization is obtained using partial fraction decomposition, which is also the case in our method, as explained below. 
A similar approach is considered in~\cite{Grimm}, where the author focus on matrices arising from stiff systems of ordinary differential equations, associated with matrices whose numerical range is negative. 
A specific rational approximation where the poles are constrained to be equidistant in a part of the complex plane is presented. These approximations are related to the functions\footnote{These functions are defined by $\varphi_\ell(x):=\sum_{k=0}\frac{x^{k}}{(k+\ell)!}$.} $\varphi_\ell$ ($\ell>0$) which do not include the exponential (which corresponds to $\varphi_0$). Moreover, the proposed approximation appears to be more efficient for large $\ell$, whereas our method is designed and efficient in the case $\ell=0$.
Various approximations of some matrix functions (including the exponential) based on rational Krylov methods with fixed denominator are presented in~\cite{guttel}. A posteriori bounds and stopping criterion in a similar framework are given in~\cite{simoncini}.

Note however that these references mainly focus on the reduction obtained by Krylov approaches and can be combined with our method. The convergence properties of Krylov methods related to matrix functions is widely documented in the literature. Among the more recent papers on this topic, we refer to~\cite{Beckermann,Diele,Druskin,Knizhnerman,Lopez}.


In the present work, we focus on the parallelization strategy associated with the rational approximation 
defined by
$R_{0,n}(z)=1/ \exp_n (-z)$ (which we simply denote by ${\mathcal R}_n(z)$ hereafter), 
where $\exp_n$ denotes the truncated Taylor series of order
$n$ associated with the exponential, i.e., $m=0$.
The poles  of ${\mathcal R}_n(z)$ 
are 
all
distinct (and well documented) allowing a partial fraction decomposition with affine
denominators.
All terms in the  decomposition are independent
hence their computation can be achieved efficiently in parallel.  

To see how these results can be used to compute matrix exponential, consider ${\mathcal E}_n$   
an approximation  of the complex exponential function depending on one parameter $n\in\N^\ast$.
This approximation naturally extends to the exponential of  
diagonal matrices by setting   ${\mathcal E}_n(\textrm{diag} (d_i)_i):=(\textrm{diag} ({\mathcal E}_n(d_i))_i$, thus to any diagonalizable matrix $A=PDP^{-1}$ by  ${\mathcal E}_n(A):=P{\mathcal E}_n(D)P^{-1}$. The latter definition is actually a property of usual matrix functions, see~\cite[Chap. 1]{higham_book}.
The approximation error for a diagonalizable matrix $A=PDP^{-1}$, with eigenvalues located in a domain of the complex plane $\Lambda$ can then be estimated as follows 
\begin{equation}
\label{equ:con}
 \|\exp(A)-{\mathcal E}_n(A)\|_2\leq \varepsilon_n \,\kappa_2(P)\,
\end{equation} 
where $\kappa_2(P)=\|P\|_2\, \|P^{-1}\|_2$ is the condition number of $P$ in the matrix norm associated with the usual Euclidean vector norm $\|\cdot\|_2$
and $\varepsilon_n:=\max_{z\in \Lambda}\,|{\mathcal E}_n(z)-\exp(z)|$. 
The approximation of $\exp(A)$ is then reduced to the approximation of the exponential on the complex plane. 
If we further assume $A$ to be Hermitian (or, more generally, normal), then $P$ is a unitary matrix and $\kappa_2(P)=1$. 
This is the case for Hermitian matrices that arise, e.g. from the space discretization of the Laplace operator, and more generally for normal matrices.  Note, however, that for an arbitrary matrix, the term $\kappa_2(P)$ may be too large and significantly deteriorate the estimate~\eqref{equ:con}. 


The paper is organized as follows. Section~\ref{sec:scal} is devoted to the approximation of the scalar exponential function.  As explained above, this approximation, denoted by ${\mathcal R}_n(z)$ is in our approach a rational function whose poles are all simple.  In Section~\ref{sec:mat}, we present the approximation of the exponential of a matrix. 
In practise, the partial fraction decomposition of ${\mathcal R}_n(z)$ 
raises some specific numerical issues related to floating-point arithmetic ; these are discussed in Section~\ref{Sec:Floating}.  
We finally demonstrate the efficiency of our method on some examples in Section~\ref{Sec:Numres}. 

\section{The scalar case}~\label{sec:scal}
For $n\in{\mathbb N}^\ast$, let us define $ \exp_n(z):=\sum_{k=0}^n\frac{1}{k!}z^k$, i.e., the exponential Taylor series truncated  at order $n$. 
It is readily seen that for all $x\in\R$ and even values of $n$, $\exp_{n}(x)$ 
is positive. 
Since $\exp_n'=\exp_{n-1}$, it follows that $\exp_n$ is strictly increasing for $n$ odd and strictly convex for $n$ even. 

\subsection{Roots of the truncated exponential series}

We denote by $(\theta^{(n)}_k)_{k=1,\cdots,n}$ the roots of the polynomial $\exp_n$.  
If $n$ is even, the roots are pairs of conjugate complex numbers and none of them is a real number. If $n$ is odd, there is one and only one real root of $\exp_n$ and the others are pairwise conjugate. Some roots of $\exp_n$ are represented on the figure \ref{fig:2m} (left panel). We see that the norm of the roots increases with $n$, which intuitively follows from the fact that the exponential function has no roots  on the whole complex plane.
However, this growth is moderate since (see \cite{marche}, for example) 
\begin{equation}
\label{equ:racn}
1\leq |\theta_k^{(n)}| \leq n. 
\end{equation}
\begin{figure}[ht]
\begin{center}
\includegraphics[width=6cm]{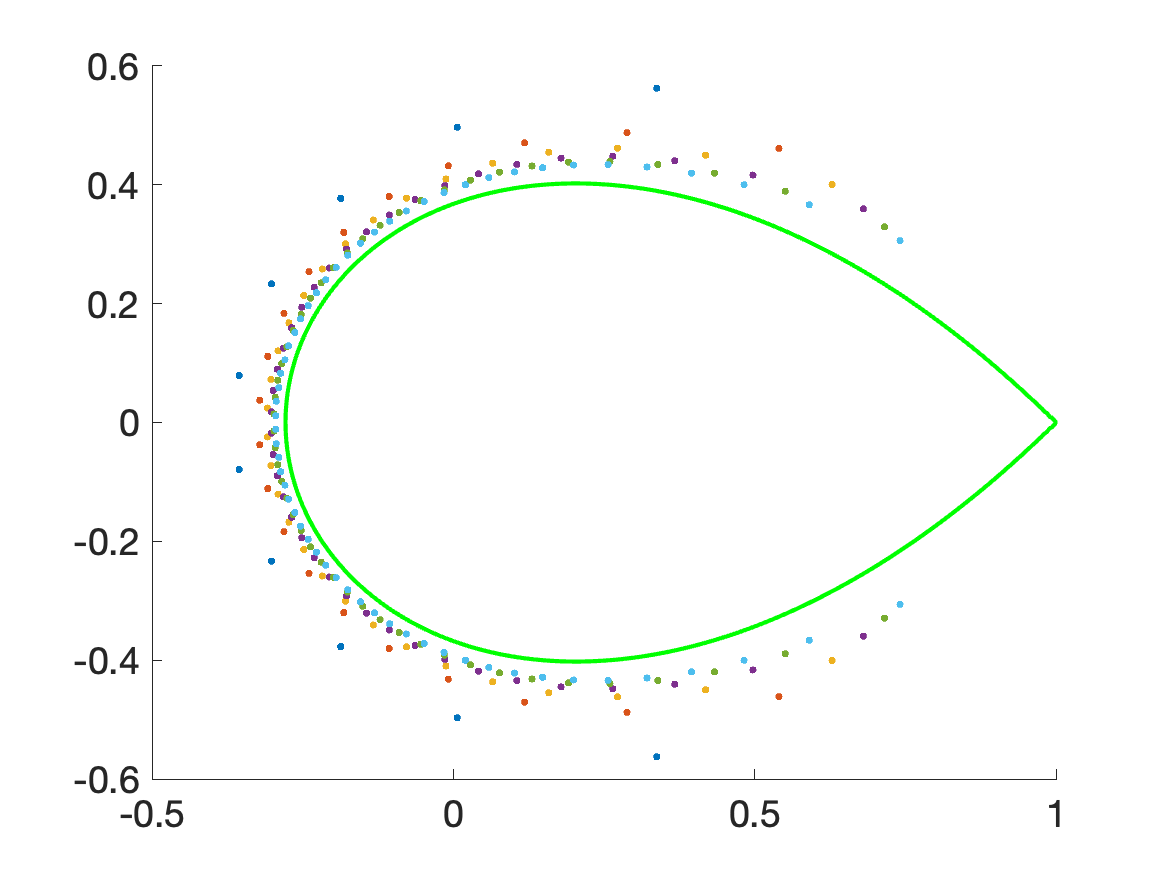}
\includegraphics[width=6cm]{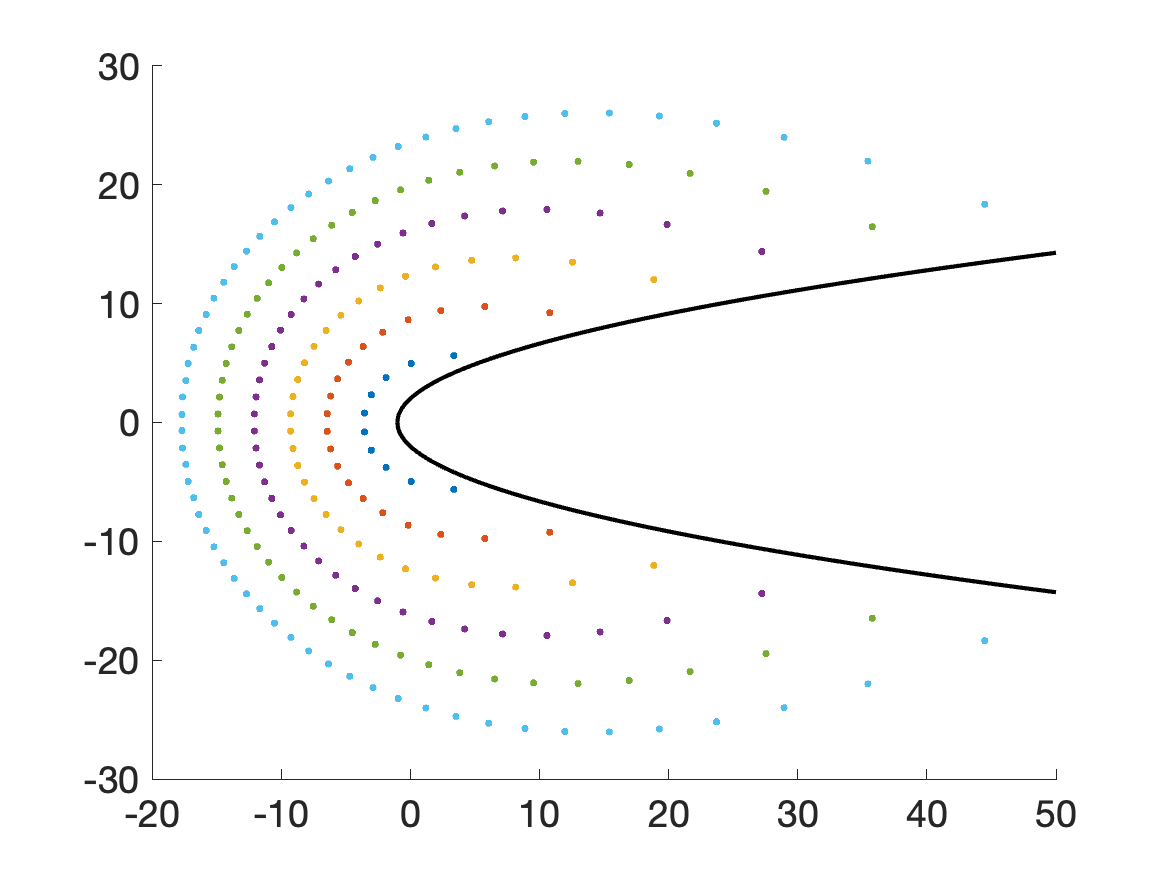}
\end{center}
\caption{Left: the roots of $z\longmapsto \exp_n(z)$, $n=10, 20, 30, 40, 50, 60$. The parabola $y^2=4(x+1)$ delimits an area containing no roots. Right: the Szeg\H{o} curve. }
\label{fig:2m}
\end{figure} 
G. Szeg\H{o} has shown in \cite{szego} that the \emph{normalized roots}, i.e., 
the roots of $\exp_n(nz)$,
approach, when $n\to\infty$, the so-called Szeg\H{o} curve, defined by 
\[\{z \in\C,\quad |ze^{1-z}| =1, \quad |z|\leq 1\}. \]
Some normalized roots and the Szeg\H{o} curve are presented in 
Figure \ref{fig:2m} (left panel). 
In view of~\eqref{equ:con}, it is interesting to determine regions of the complex plane which do not contain any roots. An example is given by the interior of the parabola of equation $y^2=4(x+1)$, which thus includes the positive real half-axis. This surprising result has been obtained by Saff and Varga in \cite{saff-varga}, see Figure \ref{fig:2m} (right panel).

\subsection{Approximation of the exponential}
We propose the following approximation of the exponential function defined for any complex number $z$ by  
\[\exp(z) \simeq {\mathcal R}_n(z) :=\frac{1}{\exp_n(-z)},\] 
which reflects the  identity $\exp(z)=\frac{1}{\exp(-z)}$.
Note that ${\mathcal R}_n(0)=1$ and that ${\mathcal R}_n$ has no real root if $n$ is even, which we will always assume in the rest of this paper.
 
This approach opens the way to a good approximation of the exponential on the half axis $(-\infty,0]$. 
We present on Figure~\ref{fig:4aaa} a graph of the exponential function, 
 its polynomial approximation $\exp_n$, and the rational approximation ${\mathcal R}_n$, on the interval $[-5,0]$.   Though the two approximations seem to fit well with the exponential function for $n=20$, we observe that for $n=10$, the rational approximation is clearly more accurate.  
\begin{figure}[ht]
\begin{center}
\includegraphics[width=6cm]{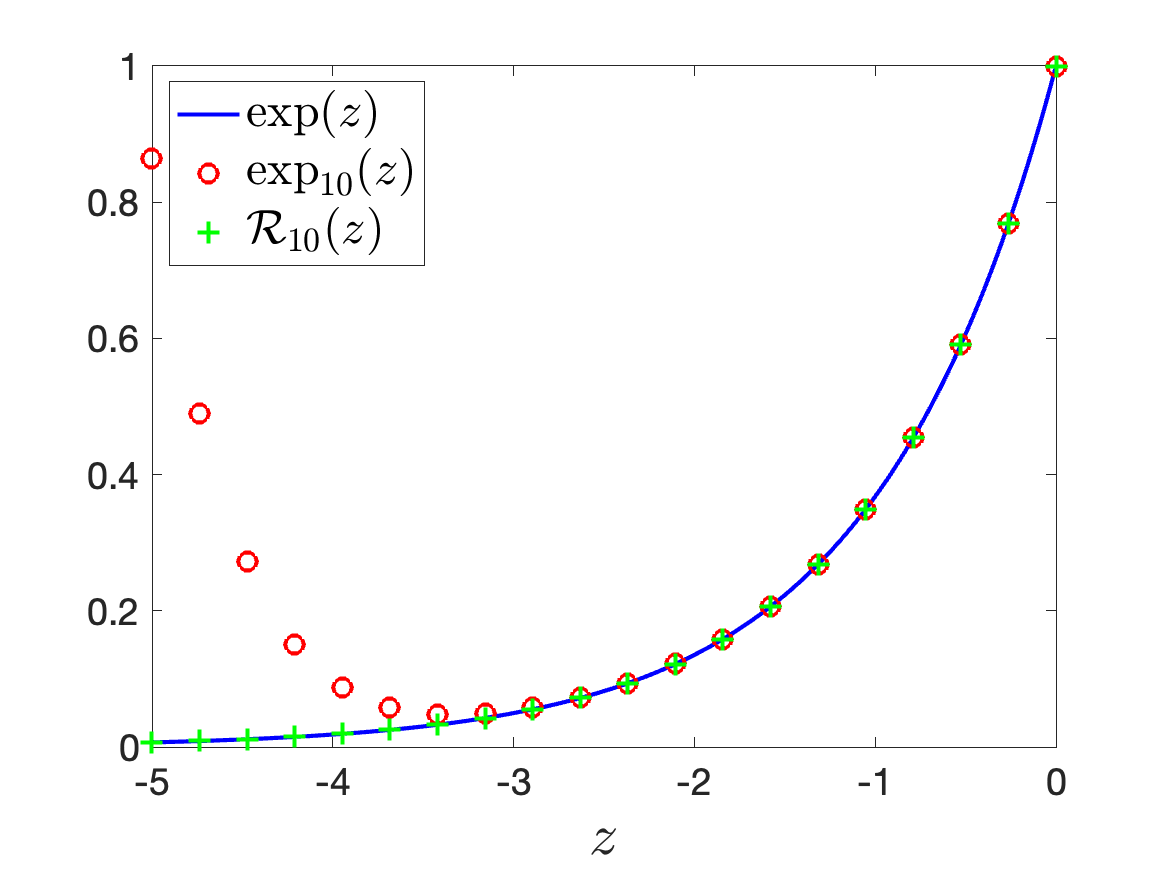}
\includegraphics[width=6cm]{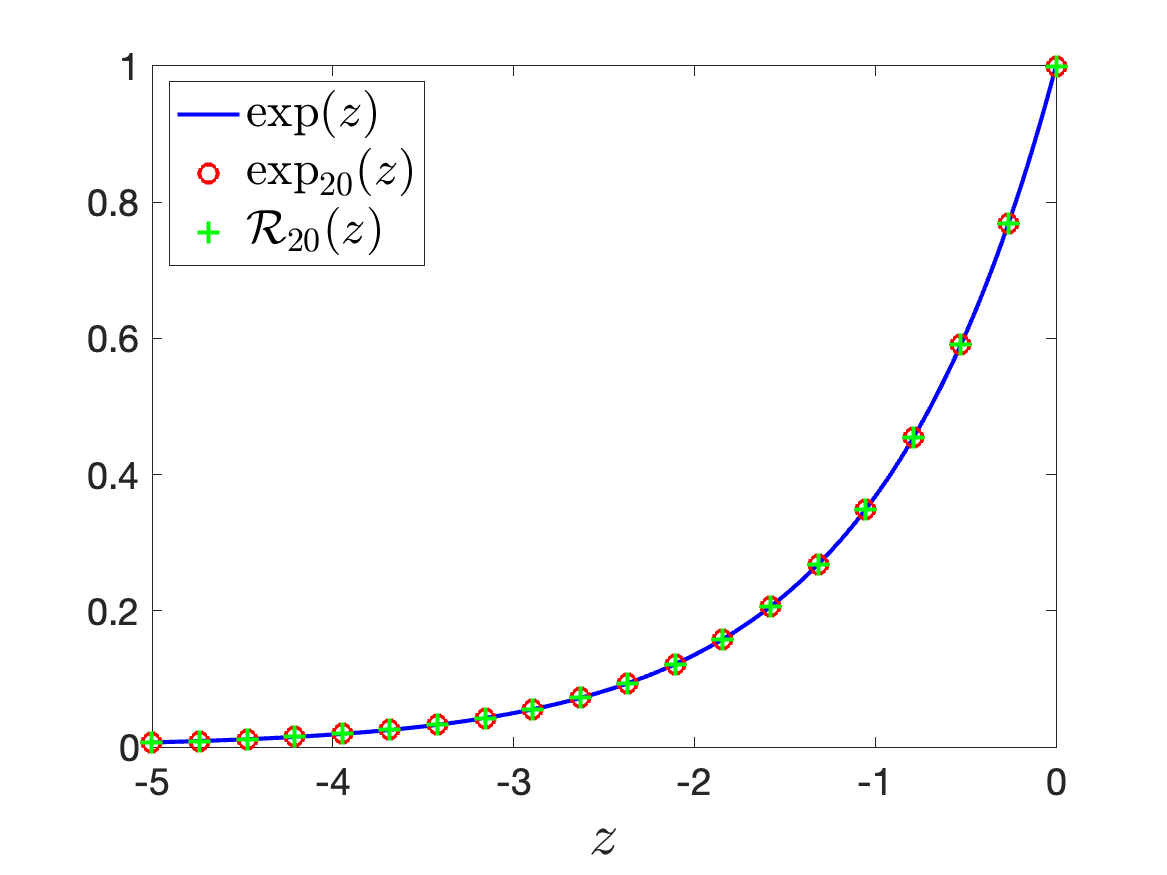}  
\end{center}
\caption{The exponential function and its polynomial and rational approximations 
rational $\exp_{n}$ and ${\mathcal R}_{n}$ on $[-5,0]$, for 
$n=10$ (left) and $n=20$ (right).} 
\label{fig:4aaa}
\end{figure}
Given $n,m\in\N$, the Pad\'e approximant~\cite{baker} of index $(m,n)$ of the exponential function is explicitly known; it is the rational function  with numerator $P_{m,n}$ and denominator $Q_{m,n}$:
\[P_{m,n}(z)=\sum_{k=0}^m\frac{(m+n-k)!m!}{(m+n)!(m-k)!k!}z^k,\quad Q_{m,n}(z)=\sum_{k=0}^n\frac{(m+n-k)!n!}{(m+n)!(n-k)!k!}(-z)^k.\] 
The function ${\mathcal R}_n$ is therefore the Pad\'e approximant of index $(0,n)$ of the exponential function, i.e., its Taylor expansion at the origin coincides with this function up to order $n$. More precisely, we have   
\begin{equation}
\label{equ:pade}
 {\mathcal R}_n^{(j)}(0)=1,\quad \quad j\in\N, \ 0\leq j\leq n,
\end{equation}
and, in the neighborhood of the origin, 
\begin{equation*}
\label{equ:dl}
{\mathcal R}_n(z) = \exp(z)+ {\cal O}(z^{n+1}).
\end{equation*}
We can slightly refine this result. Let us decompose ${\mathcal R}_n(z)$ as 
\begin{equation*}
\label{equ:serie} 
{\mathcal R}_n(z) = \exp_n(z)+ \sum_{k=n+1}^{+\infty} \,\frac{\lambda_{n,k}}{k!} z^k.
\end{equation*} 
A simple calculation shows that (recall that $n$ is supposed to be even) $\lambda_{n,n+1}=0$ and  $\lambda_{n,n+2}=-2(n+1)$. In other words, at $z=0$, the derivatives of ${\mathcal R}_n(z)$ of order greater than $n$ are not at all close to  the derivatives of the exponential function.

The partial fraction decomposition
 of ${\mathcal R}_n$ is the foundation of our numerical method to compute the exponential of a matrix. 
\begin{proposition}
 \label{pro:dec}
 We have, for all $z\in \mathbb{C}$
\begin{equation}
\label{yoyo}
{\mathcal R}_n(z)=\sum_{k=1}^n \frac{a_k^{(n)}}{z+\theta_k^{(n)}},
\end{equation}
where $\theta_k^{(n)}$ are the roots of $\exp_n$ and  
\begin{equation}
\label{yoyo2}
a_k^{(n)}= -\frac{n!}{\prod_{j\neq k}(\theta_k^{(n)}-\theta_j^{(n)})}.
\end{equation}
\end{proposition}
One should not be alarmed in the calculation of the coefficients $a_k^{(n)}$ by the relation~\eqref{yoyo2} whose denominator is a product of the differences 
$\theta_k^{(n)}-\theta_j^{(n)}$ since the difference between two roots of $\exp_n$ is uniformly lower bounded with respect to $n$ (see \cite[Theorem 4]{marche})
\begin{equation}\label{eq:ecart}
\inf_{n\geq 2} \, \min_{j\neq k}|\theta_j^{(n)}-\theta_k^{(n)}|\geq \gamma :=0.29044\cdots    
\end{equation}
thus avoiding to divide by too small numbers in~\eqref{yoyo2}. Note also that other expressions can be used to compute the coefficients $a_k^{(n)}$, e.g., 
\begin{equation}
\label{yoyo234}
a_k^{(n)}= \frac{-1}{\exp_n'(\theta_k^{(n)})}=\frac{ n!}{(\theta_k^{(n)})^n}.
\end{equation} 
Numerical properties of these formula are investigated in Section~\ref{Sec:Floating}.
\subsection{Convergence and error estimate}
The rational approximation of a real or complex function is a well-documented problem. Results concerning  existence of best approximation, uniqueness, computation can be found, e.g., in~\cite{meinar}. 
In this section, we focus on the convergence properties of ${\mathcal R}_n(z)$ on $\left] -\infty, 0\right]$. This interval includes the spectrum of negative Hermitian matrices that we consider in this paper. 
Note however that results on others domains are available.  
Indeed,  let ${\mathbb P}_{k}$ denotes the set of polynomials of degree at most $k$, ${\mathbb P}_{m,n}$ the set of rational functions $p/q$, $p\in{\mathbb P}_{m}$, $q\in{\mathbb P}_{n}$, and 
\[E_{m,n}^{}(\exp ,\Lambda):=\min_{r\in{\mathbb P}_{m,n}(\Lambda)}\,\max_{z\in \Lambda}\,|\exp(z)-r(z)|,\]
the error of best uniform approximation of the exponential  on $\Lambda\subset\C$. 
On bounded domains, we have for example~\cite{braess}  
\[ E_{m,n}^{}(\exp ,[-1,1])\underset{n+m\to \infty}{\simeq} \, \frac{n!m!}{2^{n+m}(n+m)! (n+m+1)!}.\]
The analog of this result on the disk $B(0,\varrho)$ of radius $\varrho$ centered in $0$ is~\cite{tre} 
\[ E_{m,n}^{}(\exp ,B(0,\varrho) )\underset{n+m\to \infty}{\simeq} \, \frac{n!m!\varrho^{n+m+1}}{(n+m)! (n+m+1)!}.\]
In both  cases, we observe a fast decreasing of the error when $\Lambda$ is bonded. This is not the case in the applications which have 
motivated our study. The case $\Lambda= \left ]-\infty, 0\right ]$ is discussed in the 
pioneering work of \cite{varga}.
The authors show that the best approximation error $E_{0,n}^{}(\exp,]-\infty, 0])$ decays linearly and exhibit a particular function, which just happens to be our rational approximation ${\mathcal R}_n$.
\begin{proposition}(\cite[Lemma 1]{varga})
\label{pro:err}
We have for any real $x\leq 0$
\begin{equation}
\label{equ:error}
|{\mathcal R}_n(x)-\exp(x) | \,\leq \, M_1(n):=\frac{1}{2^n}.
\end{equation}
\end{proposition}

The convergence of ${\mathcal R}_n(x)$ to $\exp(x)$ is therefore linear on the real negative half-line\footnote{For the sake of completeness, let us note that the optimal linear decrease is given by Sch\"onhage in~\cite{scho} who showed that $\lim_{n \to + \infty}\, E_{0,n}^{1/n}(\exp,]-\infty, 0])=1/3$.}.
The error on this interval is represented in Figure \ref{fig:ErrC10} for various values of $n$. 
Figure \ref{fig:ErrC1} shows iso-curves of the norm of the error for $n=32$ as well as  points $-\theta_k^{(n)}$.  The values presented in the graphics for $n=32$ are the $10^{-k}$ for $k=0,\cdots, 14$. One observes there the rapid decay of the approximation along the real half-axis. We also observe a remarkable decay in the whole left half-plane. 
\begin{figure}[ht]
\begin{center}
\includegraphics[width=6cm]{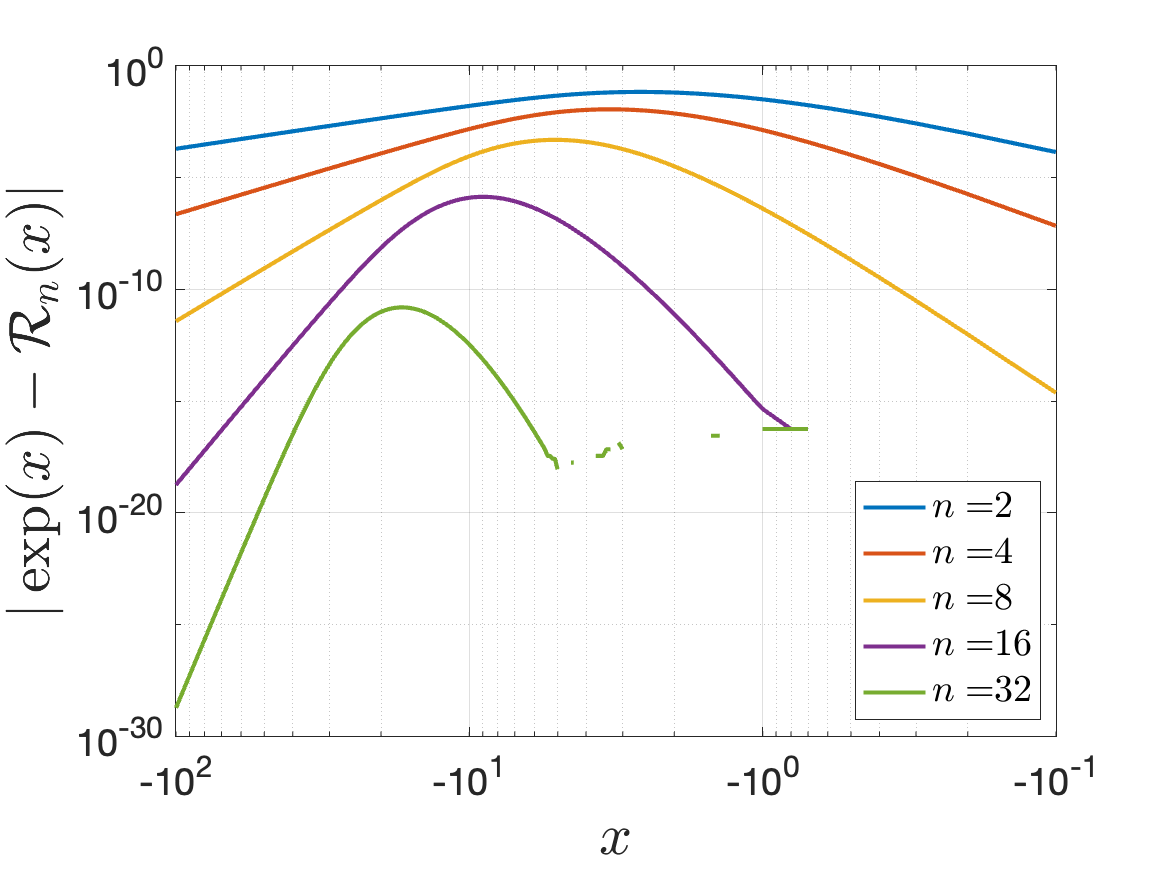}  
\end{center}
\caption{The error $|{\mathcal R}_{n}(x)-\exp(x)|$ for $n\in\{4, 8, 16, 
32\}$.} 
\label{fig:ErrC10}
\end{figure} 
  
\begin{figure}[ht]
\begin{center}
\includegraphics[width=6cm]{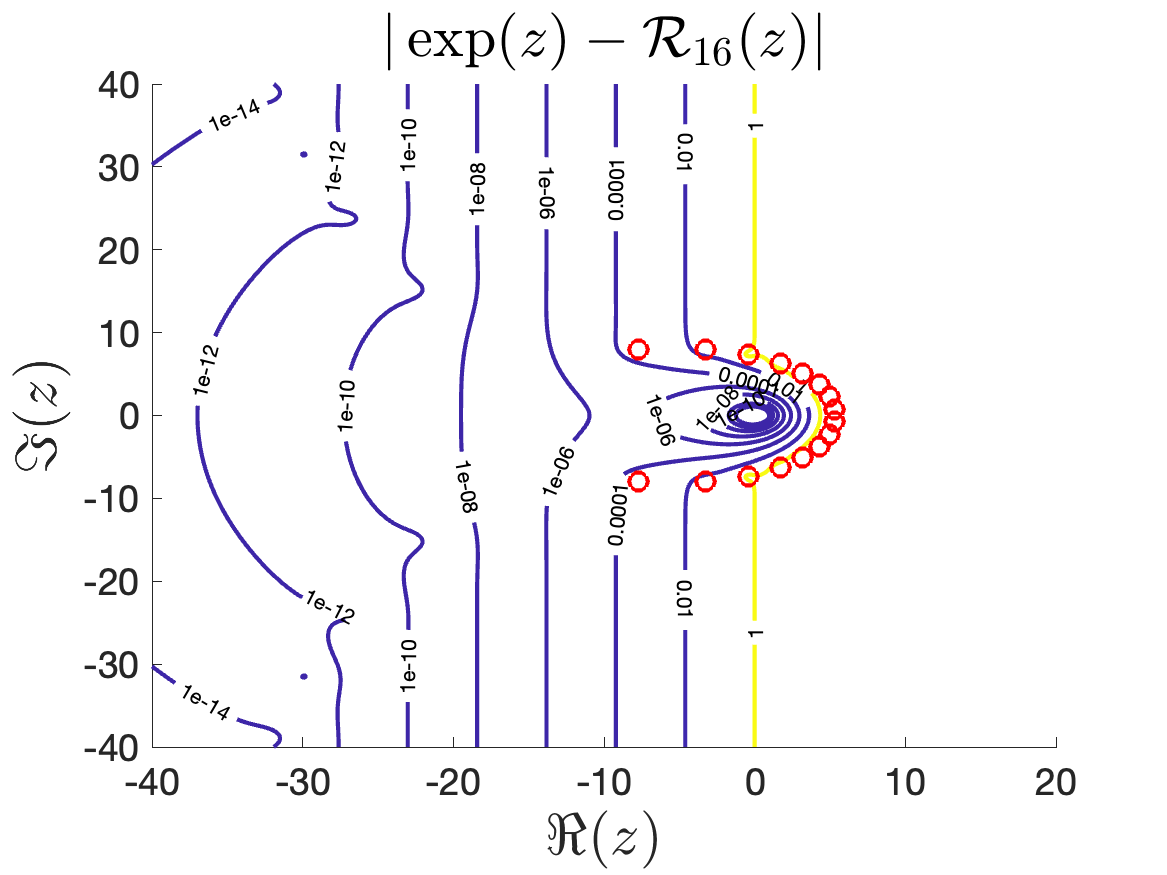}
\includegraphics[width=6cm]{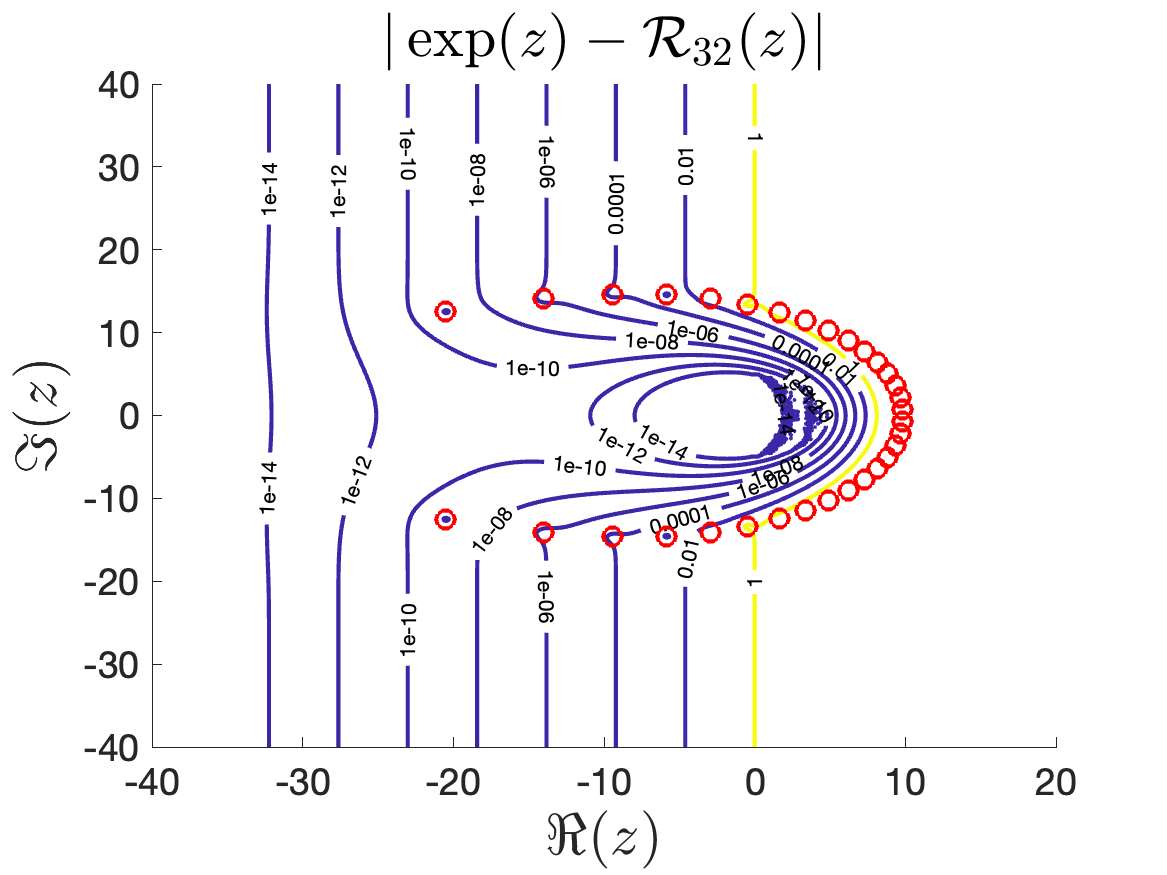}  
\end{center}
\caption{ Norm of error ${\mathcal R}_{n}(z)-\exp(z)$ and poles of ${\mathcal R}_{n}$ for $n=16$ (left) and $n=32$ (right).} 
\label{fig:ErrC1}
\end{figure}
Before going further, let us state a technical result.

\begin{lemma}\label{lem:pre}
The function $f_n(x):=\exp_n(x)\exp(-x)$ satisfies
\begin{align}
f_n(n+1)< & \, \frac 12 \label{ineq:12}\\
f^2_n(x)\leq & \, \frac{n+1}{n}f_{n-1}(x)f_{n+1}(x).\label{ineq:fn}
\end{align}
\end{lemma}  
\begin{proof}
We first show that 
\begin{equation}\label{eq:exp12}
\exp_{n}(n+1)=\sum_{k=0}^{n}\frac{(n+1)^k}{k!}< \exp(n+1) -\exp_{n}(n+1)= \sum_{k=n+1}^{+\infty}\frac{n^k}{k!},
\end{equation}
which directly leads to~\eqref{ineq:12}.
To get~\eqref{eq:exp12}, we compare the terms $k=n-j$ and $k=n+j-1$ of the respective sums.
Precisely, we have 
\[\frac{n^{n-j}}{(n-j)!} \leq \frac{n^{n+j-1}}{(n+j-1)!}.\]
The proof is by induction on $j$. For $j=1$, we actually have equality. Assuming that the property is true at rank $j$, we have
\begin{align*}
\frac{n^{n-(j+1)}}{(n-(j+1))!} =& \frac{n^{n-j}}{(n-j)!}\frac{n-j}{n} < \frac{n^{n+j-1}}{(n+j-1)!}\frac{n-j}{n} 
= \frac{n^{n+j}}{(n+j)!}(1-(\frac j n)^2)\\
<& \frac{n^{n+j}}{(n+j)!},
\end{align*}
hence the result. Inequality~\eqref{ineq:fn} simply follows from the Cauchy-Schwarz inequality applied to  
\[\int_x^{+\infty} \exp(-t)t^n dt = \int_x^{+\infty} (\exp(-t/2)t^{\frac{n-1}{2}}) 
\, 
(\exp(-t/2)t^{\frac{n+1}{2}})\,dt,\]
since $f_n(x)= \int_x^{+\infty} \exp(-t)\frac{t^n}{n!} dt$.
\end{proof}

In the following proposition, we summarize some properties of the function $err_n: x \in \left]- \infty,0\right]\longmapsto 
{\mathcal R}_{n}(x)-\exp(x)>0$.

\begin{proposition}
\label{pro:wxcv1} 
For $n\geq 1$, the function $err_n$ reaches its maximum at a single point $\xi_n <0$, is 
increasing on  $]-\infty, \xi_n]$ and 
decreasing on $[\xi_n,0]$. Moreover, we have
\begin{equation}\label{ineq:loc}
\frac n 2 \leq -\xi_n \leq n+2.
\end{equation}

\end{proposition}
\begin{proof}
To simplify the notation, we introduce $y=-x$ and study the error $err_n(y)={\mathcal R}_{n}(-y)-\exp(-y)$ on the half-axis $(0, +\infty[$, in which we have excluded $y=0$ where the error cancels.
Since 
\begin{align}err_n'(y)
  = & \exp(2y) \frac{(\exp_n(y)\exp(-y))^2-\exp_{n-1}(y)\exp(-y)}{\exp_n(y)^2} \label{eq:sign}\\
  = & y^n\frac{\exp(-y) }{\exp_n(y)^2} \left(\frac{1}{n!}\exp_n(y)-\exp_{n-1}(y)\frac{\exp(y)-\exp_n(y)}{y^n}\right), \nonumber  
\end{align}
the variations on $err_n$ are  determined by the sign of 
\begin{equation}\label{eq:defgn}
g_n(y):= \frac{1}{n!}\exp_n(y)-\exp_{n-1}(y)\frac{\exp(y)-\exp_n(y)}{y^n}.
\end{equation}
In this formula, the last term as well as all its derivatives is positive on $I=]0, +\infty[$ so that
$g_n^{(n+1)}(y)< 0$ on this interval.
To  prove that $g_n$ has an unique zero $\xi_n$ in $I$, we shall show that for some $a_1>0$, $g_n'$ is strictly positive on 
in interval $]0,a_1[$ and strictly negative on $]a_1,+\infty[$. That guarantees the result since $g_n(0)>0$ and $\lim_{y\rightarrow+\infty}g_n(y)=-\infty$. The values of $g_n^{(k)}$ $(k=1,\cdots,n)$ on both sides of interval $I$ are of importance in the analysis. Note first that $\lim_{y\rightarrow+\infty}g_n^{(k)}(y)=-\infty$ for all $k\in \N$ and that  
the sequence $u_k:=g_n^{(k)}(0)= \frac{1}{n!} - \sum_{\ell=1}^k \frac{k!}{(k-\ell)!(n+\ell)!} $ is decreasing. Indeed, for $k=1,\cdots,n-1$ we have 
\begin{align*}
u_{k+1}-u_k
=& -\frac{(k+1)!}{(n+k+1)!} + \sum_{\ell=1}^k \frac{k!}{(k-\ell)!(n+\ell)!} -  \frac{(k+1)!}{(k+1-\ell)!(n+\ell)!} \\
=& -\frac{(k+1)!n!}{(n+k+1)!} -k!\sum_{\ell=1}^k \frac{1}{(k-\ell)!(n+\ell)!}\frac{\ell}{k+1-\ell} <0.
\end{align*}
Since the first term in the right hand side of~\eqref{eq:defgn} is a polynomial of order $n$, we have $g_n^{(n+1)}(y)< 0$ on $I$.
Hence $g_n^{(n)}$ is strictly decreasing on this interval.
If $u_n>0$, then for some $a_n>0$, $g_n^{(n)} >0$  on some interval $[0, a_n[$ and $g_n^{(n)}<0$ on $]a_n, +\infty[$. 
Hence function $g_n^{(n-1)}$ is strictly increasing on the first interval and decreasing on the second one. It follows that there exists a unique $a_{n-1}>a_n$ where this function vanishes. This property clearly spreads to $g_n$.
Otherwise, $u_n\leq 0$, and the previous reasoning can be applied to the largest $n'$ such that $u_{n'}> 0$.

To prove~\eqref{ineq:loc}, we first show that $err'_n(n/2)<0$, i.e., $g_n(n/2)>0$. Set $y=n\theta$, with $0<\theta<1$. We have
\begin{align*}
\frac{\exp(y)-\exp_n(y)}{y^n}
=&\frac 1{(n\theta)^n}\sum_{\ell = n+1}^{+\infty} \frac{\theta^\ell}{\ell !/n^\ell} \\
\leq&\frac 1{(n\theta)^n}\frac{n^{n+1}}{(n+1) !}\sum_{\ell = n+1}^{+\infty} \theta^\ell  
=\frac{n}{(n+1) !}\frac{\theta}{1- \theta}.
\end{align*}
As a consequence, we get
\begin{align*}
g_n(n\theta)
\geq & \frac{(n\theta)^n}{(n!)^2} + \frac{\exp_{n-1}(n\theta)}{n!}\left(1-\frac{n}{n+1}\frac{\theta}{1- \theta}\right),
\end{align*}
which is positive when $\theta=1/2$. 

We then prove that $err'_n(n+2)<0$. Rewriting~\eqref{eq:sign} with the notation of Lemma~\ref{lem:pre}, we get $err_n'(y)=\exp(2y) \frac{(f_n(y))^2-f_{n-1}(y)}{\exp_n(y)^2}$. The task is now to find the sign of $(f_n(n+2))^2-f_{n-1}(n+2)$. Because of  Lemma~\ref{lem:pre}, we have
\begin{align*}
(f_n(n+2))^2-f_{n-1}(n+2)\leq & \left( \frac{n+1}{n}f_{n+1}(n+2)-1\right)f_{n-1}(n+2)\\
\leq & \left( \frac{n+1}{2n}-1\right)f_{n-1}(n+2),
\end{align*}
where the former inequality follows from~\eqref{ineq:fn} and the latter from~\eqref{ineq:12}. This shows that $(f_n(n+2))^2-f_{n-1}(n+2)<0$. Hence $err'_n(n+2)<0$.
\end{proof}
\section{Approximation of the exponential of Hermitian matrices}~\label{sec:mat}
Let $A$ be a square matrix of ${\cal M}_d(\mathbb{C})$. Given $n>1$, we suppose that all matrices $A-\theta_k^{(n)}I$ are invertible, i.e., their spectrum does not contain any root of any $\exp_n$. This is the case if the matrix $A$ is Hermitian (recall that $n$ is supposed to be even). The same is true for any matrix provided that $n$ is large enough. We propose the following approximation of the exponential of $A$ 
\begin{equation}
\label{equ:ExpoMatRat}
\exp(A) \simeq {\mathcal R}_n(A):=\sum_{k=1}^na_k^{(n)}(A+\theta_k^{(n)}I)^{-1},
\end{equation}
where $I$ denotes the identity matrix.
\begin{remark}
Note that ${\mathcal R}_n(0)=I$ and that if the matrix $D\in{\cal M}_d(\mathbb{C})$ is diagonal, so is matrix ${\mathcal R}_n(D)$ with $({\mathcal R}_n(D))_{i,i}= {\mathcal R}_n(D_{i,i})$.
On the other hand, for any invertible matrix $P\in{\cal M}_d(\mathbb{C})$, we have 
\[{\mathcal R}_n(PAP^{-1})=P{\mathcal R}_n(A)P^{-1}.\] 
 
\end{remark}


From now on, we restrict our attention to negative Hermitian matrices.  In view of   Proposition~\ref{pro:wxcv1}, we can state a
specific estimate in this case. 
\begin{theorem}
Assume that $Spec{(A)}\subset \mathbb{R}^-$ and let $\varrho(A):=\max_{\lambda\in Spec{(A)}} |\lambda|$. If $n>2\varrho(A)$, 
then 
\[\|\exp(A)-{\mathcal R}_n(A)\|_2\leq \varepsilon,\] 
where $\varepsilon={\mathcal R}_n(-\varrho(A))-\exp(-\varrho(A))$.
\end{theorem}
\begin{remark}[Shifting method for nonnegative Hermitian matrices]
\label{rem:shift}
Since the spectrum of a real-valued matrix can be localized everywhere in the complex plane, we cannot guarantee that~\eqref{equ:error} holds in the general case. This problem can be solved by a shifting method in the case of Hermitian matrices. Let $A$ be 
an Hermitian matrix and $c\in{\mathbb R}$ a bound of its spectrum, $c\geq \alpha(A):=\max_i \lambda_i$. 
Since $Spec(A-cI)\subset \R^-$
, we can consider the approximation 
\[\exp(A) =e^{c}\exp(A-cI) \simeq e^{c}\,{\mathcal R}_n(A-c I).\]
But the term $e^c$ can be very large so that the approximation is only relevant for moderate values of $c$. 
However, we have
\[\frac{\|\exp(A) - e^{c}\,{\mathcal R}_n(A-c I)\|_2}{\|\exp(A)\|_2}\leq \frac{e^c}{\|\exp(A)\|_2} \|\exp(A-cI) -{\mathcal R}_n(A-c I)\|_2.\]
Assuming that $A$ is  Hermitian, we have $\|\exp(A)\|_2=e^{\alpha(A)}$, so that the relative error can be controlled in this case. Some numerical results about this strategy are presented in Section~\ref{Sec:Numres}.
\end{remark} 


Many applications require in practice to compute a matrix-vector product instead of assembling the full matrix. In such a case, given $v\in{\mathbb C}^d$,  $y=\exp(A)v$ is computed by 
\begin{equation}
\label{equ:expAx}
y  \simeq {\mathcal R}_n(A)v=\sum_{k=1}^na_k^{(n)}y_k^{(n)},
\end{equation}
with $y_k^{(n)}$ the solution to the linear system $(A+\theta_k^{(n)}I)y_k^{(n)}=v$. Each $y_k^{(n)}$ could be computed separately from the others leading to significant savings in computing time as illustrated by our numerical tests, see Section~\ref{Sec:Numres}.

\section{Floating-point arithmetic and numerical implementation}\label{Sec:Floating}
In this section, we examine the efficiency of the approximation 
\[\exp(x)\simeq \sum_{k=1}^n \frac{a_k^{(n)}}{x+\theta_k^{(n)}},\]
where $x$ is assume to be a real number. We decompose the error according to
\[\underbrace{\exp(x)-\sum_{k=1}^n \frac{a_k^{(n)}}{x+\theta_k^{(n)}}}_{e_1(x)}=\underbrace{
\left(\exp(x)-\frac{1}{\exp_n(-x)}\right)}_{e_2(x)} +\underbrace{\left(\frac{1}{\exp_n(-x)}-\sum_{k=1}^n \frac{a_k^{(n)}}{x+\theta_k^{(n)}}\right)}_{e_3(x)}.\]
\begin{figure}[ht]
\begin{center}
\includegraphics[width=6cm]{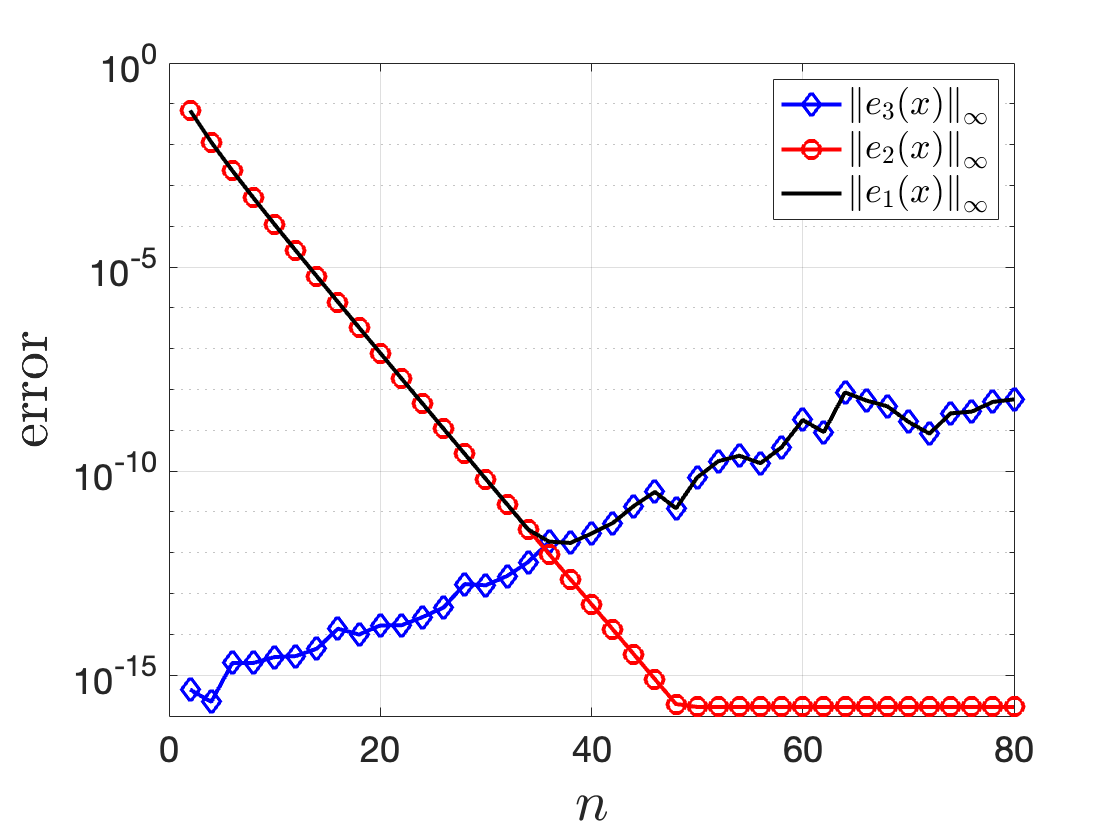} 
\includegraphics[width=6cm]{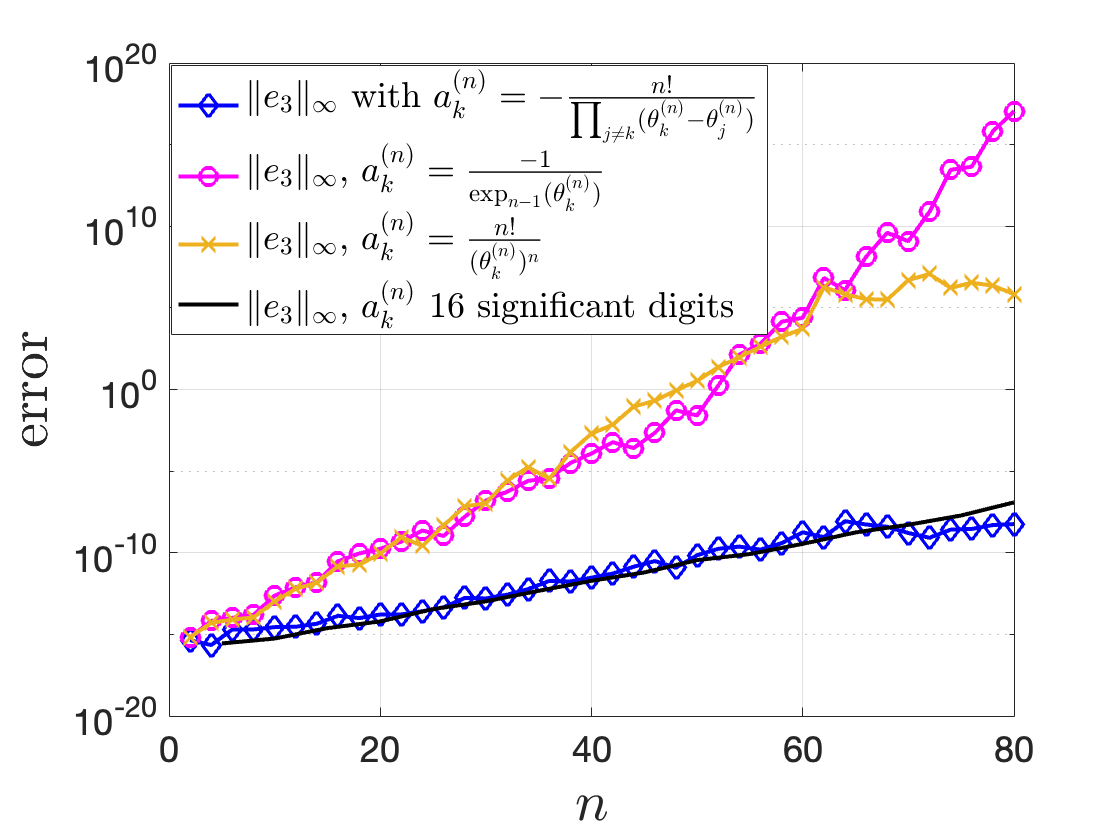} 
\end{center}
\caption{Left: uniform norm of $e_1$, $e_2$, and $e_3$ over $[-100,0]$  as a function of $n$. Right: norm of $e_3$ using various definitions of $a_k^{(n)}$, the black curve computation is done with $a_k^{(n)}$ computed with $32$ significant digits and truncated to $16$ significant digits.}
\label{fig:ineg_triang}
\end{figure}
The latter cancels 
in exact arithmetic. However, working for example in a finite precision of $16$ significant digits, we see 
on Figure \ref{fig:ineg_triang} (left panel) that 
in practise $e_1$ decreases until approximately $n = 34$ and then increases. This observation shows that our approximation is 
%
not relevant 
in practice for large values of $n$. 
This behavior can be explained by an analysis of $e_2$ and $e_3$. The former decreases 
as predicted by Proposition~\ref{pro:err}.
The latter increases with respect to $n$. 
The increase in 
$e_3$ is related to the partial fraction decomposition in floating-point arithmetic 
which is our framework hereafter, since we use MATLAB~\cite{MATLAB:R2021b} and Octave~\cite{Octave} with double precision. The accuracy actually 
deteriorates for larger values of $n$. Indeed, the three  equivalent definitions of the coefficients $a_k^{(n)}$ given by~\eqref{yoyo2} and~\eqref{yoyo234} lead in practise to different numerical results. The uniform norm of $e_3$ obtained with each of these definitions is shown on Figure \ref{fig:ineg_triang} (right panel). The formula given in~\eqref{yoyo2} gives the most precise result, which is actually very similar to the one obtained by keeping the exact $16$ first digits. Hence, we  use~\eqref{yoyo2}  in the sequel. 

In order to understand the 
influence of the floating-point arithmetic, we give in the next proposition 
a bound 
which guarantees 
a certain precision for a given $n$ when working with a floating-point arithmetic of $D$ significant digits. 
 

\begin{proposition}\label{prop:digit}
Denote by $\widetilde{a}_k^{(n)}$ and $\widetilde{\theta}_k^{(n)}$, the $D$-significant digits approximations of $a_k^{(n)}$ and $\theta_k^{(n)}$, and assume that
\begin{equation}\label{eq:condDn}
\gamma > n10^{(1-D)}  
\end{equation}
with $\gamma$ defined in (\ref{eq:ecart}). 
We have the following upper bound, for $x \in \mathbb{R}$ :
\begin{equation}\label{eq:upper_bound}
    \left| \frac{1}{\exp_n(-x)}-\sum_{k=1}^n \frac{\widetilde{a}_k^{(n)}}{\widetilde{\theta}_k^{(n)} + x}\right| \leq M_2(n,D):=\left(C_1(D) + C_2(n,D)\right)\sum_{k=1}^n \left|\widetilde{a}_k^{(n)}\right|,
\end{equation}
where
\[C_1(D) := \frac{2\cdot10^{(1-D)}}{\gamma\left(1 - 10^{(1-D)}\right)}, \quad C_2(n,D) := \frac{4n\cdot10^{(1-D)}}{\gamma(\gamma - n\cdot10^{(1-D)})}.\]
\end{proposition}
Note 
that the condition given in~\eqref{eq:condDn} 
is not restrictive: it holds for example in the case of $16$ significant digits even in the case where $n \approx 10^{10}$.

This result is illustrated in Figure~\ref{fig:majorant_chiffsign}. 
We see in this example that with $16$ significant digits, the bound obtained in~\eqref{eq:upper_bound} implies that working with $n=30$  guarantees an error of order 
 $10^{-8}$ and get an actual order of $10^{-10}$. We see however that to increase the accuracy, we could work up to $n\approx 36$. 
\begin{figure}[ht]
\begin{center}
\includegraphics[width=6cm]{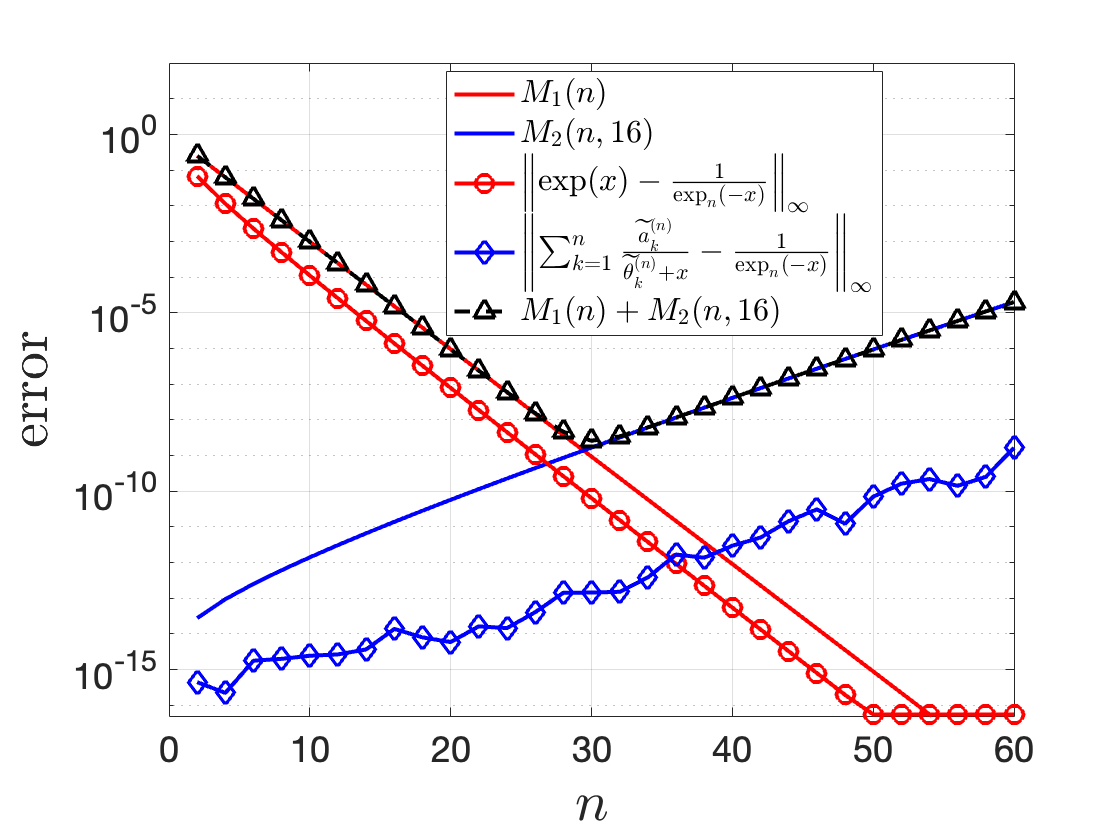}
\end{center}
\caption{Uniform norm of $e_2$ and $e_3$ over $[-100,0]$, in $16$ significant digits, and upper bounds : $M_1(n) := 1/2^n$ (see Proposition~\ref{pro:err}) and $M_2(n,D=16)$ given by (\ref{eq:upper_bound}).}
\label{fig:majorant_chiffsign}
\end{figure}
\begin{proof}
(Proposition~\ref{prop:digit})
Since we are dealing with numerical approximations based on $D$ significant digits, we 
consider the first  $D$ digits of $\widetilde{a}^{(n)}_k$ and $\widetilde{\theta}^{(n)}_k$ to be exact.
Then, for any complex number $z$ and its approximation $\widetilde{z}$ 
we have:
\begin{equation}\label{eq:digit_approx}
    \widetilde{z} = z(1 + \varepsilon_z), \quad \left|\varepsilon_z\right| \in \left[10^{-D}, 10^{(1-D)}\right].
\end{equation}
Writing $r^{(n)}_k(x) = \frac{1}{\theta^{(n)}_k + x}$ and $\widetilde{r}^{(n)}_k(x)= \frac{1}{\widetilde{\theta}^{(n)}_k + x}$, we see that finding an upper bound for the left side of (\ref{eq:upper_bound}) amounts to finding 
an upper bound for:
\begin{equation*}
    \begin{split}
    \sum_{k=1}^n r^{(n)}_k(x)\left(\widetilde{a}_k^{(n)} - a_k^{(n)}\right) 
   + \widetilde{a}_k^{(n)}\left(\widetilde{r}^{(n)}_k(x) - r^{(n)}_k(x)\right)
   \\ =  
    \sum_{k=1}^n r^{(n)}_k(x) a_k^{(n)} \left( \frac{\varepsilon_{a_k^{(n)}}}{1 + \varepsilon_{a_k^{(n)}}} \right)
     + 
    \sum_{k=1}^n \widetilde{a}_k^{(n)}\left(\widetilde{r}^{(n)}_k(x) - r^{(n)}_k(x)\right),
\end{split}
\end{equation*}
where the equality follows from (\ref{eq:digit_approx}).
Combining~\eqref{equ:racn} with~\eqref{eq:digit_approx}, we get $|\widetilde{\theta}_k^{(n)} - \theta_k^{(n)}| = |\varepsilon_{\theta_k^{(n)}}\theta_k^{(n)}| \leq n\cdot 10^{(1-D)} $. 
Combining~\eqref{eq:ecart} with the fact that
for $n$ even, $\theta^{(n)}_k$ are strictly not real, we obtain that 
$\left|\theta_k^{(n)} + x\right| \geq \left|\mathcal{I}m(\theta_k^{(n)})\right| \geq \frac{\gamma}{2}$
when $x \in \mathbb{R}$. 
As a consequence $\left|r^{(n)}_k(x)\right| \leq \frac{2}{\gamma}$ for all $x \in \mathbb{R}$. 
In the same manner, we can see that 
\[
\left|\widetilde{\theta}_k^{(n)} + x\right| \geq \left|\mathcal{I}m(\widetilde{\theta}_k^{(n)})\right| \geq \left||\mathcal{I}m({\theta_k^{(n)}})| - |\mathcal{I}m(\varepsilon_{\theta_k^{(n)}}{\theta_k^{(n)}})|\right| \geq \frac{\left|\gamma - n10^{(1-D)}\right|}{2}
\]
which follows from~\eqref{eq:condDn}. Consequently, $\left|\widetilde{r}^{(n)}_k(x)\right| \leq \frac{2}{\gamma - n10^{(1-D)}}$ for all $x \in \mathbb{R}$.

Finally, we have $|\widetilde{r}^{(n)}_k(x) - r^{(n)}_k(x)| = |\widetilde{r}^{(n)}_k(x)||r^{(n)}_k(x)||\widetilde{\theta}_k^{(n)} - \theta_k^{(n)}|$ so that $|\widetilde{r}^{(n)}_k(x) - r^{(n)}_k(x)| \leq \frac{4n\cdot 10^{(1-D)}}{\gamma(\gamma - n10^{(1-D)})}$. 
Combining all theses inequalities with 
$|\varepsilon_{a_k^{(n)}}| \leq 10^{(1-D)}$, and $\left|\frac{\varepsilon_{a_k^{(n)}}}{1 + \varepsilon_{a_k^{(n)}}}\right| \leq \frac{10^{(1-D)}}{1 - 10^{(1-D)}}$, we get (\ref{eq:upper_bound}).
\end{proof}

The graphs of $e_1(x)$ for $x=-10$ obtained with various number of significant digits is given in Figure  \ref{fig:maple}. 
We see that the larger the number of significant digits, the later $e_1$ starts increasing. 
It follows that floating-point arithmetic precision must be adapted to $n$ which is in practise the number of processor used in the computation. 

\begin{figure}[ht]
\begin{center}
\includegraphics[width=6cm]{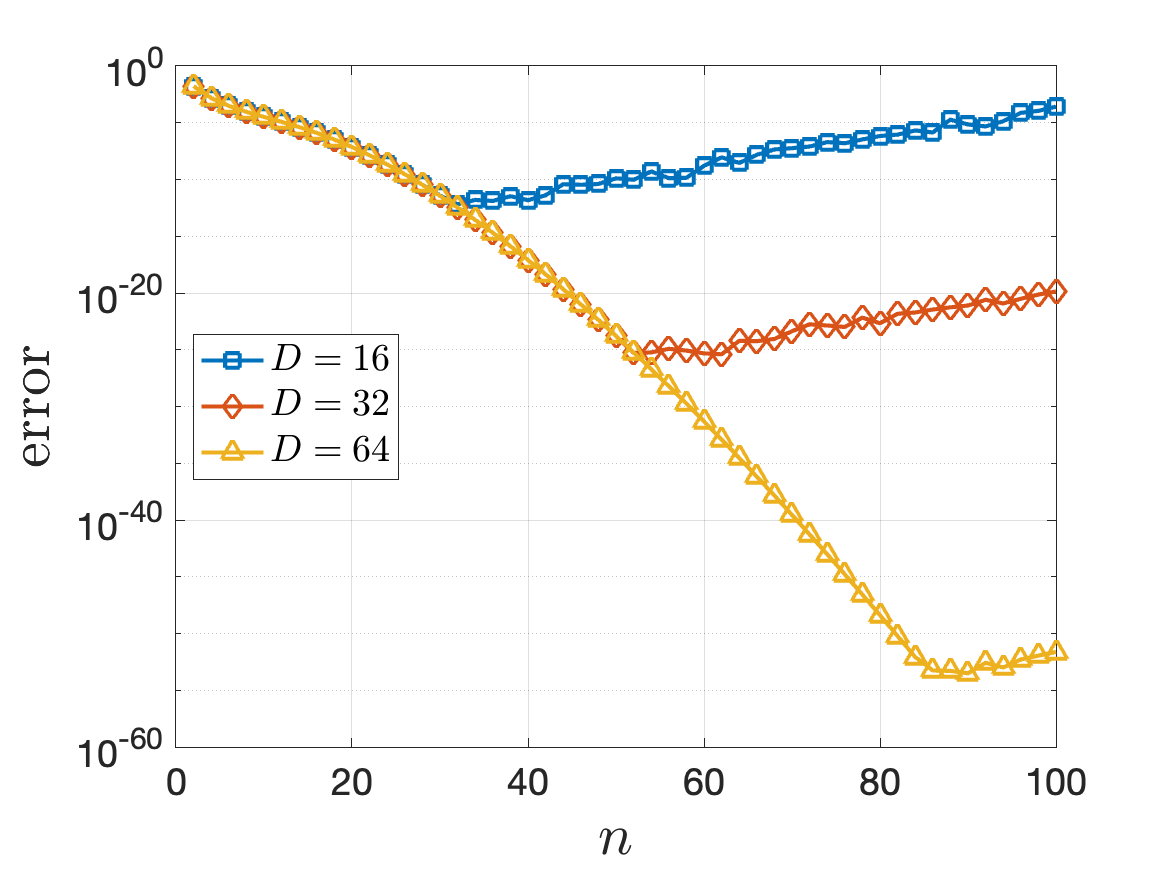}
\end{center}
\caption{Error ($e_1$) on the approximation of $\exp(x)$ vs $n$ ($x = -10$) computed with Maple, with $16$, $32$ and $64$ significant digits, respectively.}
\label{fig:maple}
\end{figure}  

\section{Numerical efficiency}\label{Sec:Numres} 

In this section, we test the performance of our method on MATLAB and Octave, and compare it with several other algorithms: the \texttt{expm} functions available in these softwares, and a rational Krylov method. Recall that \texttt{expm} is based on the combination of a Padé approximation with a scaling and squaring technique. MATLAB uses the variant described in~\cite{highS} and~\cite{high} whereas Octave uses the variant described in~\cite{ward}.
\begin{remark}\label{rem:two_time_real_part}
If $n$ is even and $x \in \mathbb{R}$, we can compute twice as fast ${\mathcal R}_n(x)$. Indeed, the complex numbers $\theta^{(n)}_k$ are in this case a set of conjugate pairs as well as $a^{(n)}_k$, and $\frac{1}{\theta^{(n)}_k + x}$. 
Assuming that the labelling is such  that $\theta^{(n)}_{2\ell+1} = \overline{\theta^{(n)}_{2\ell}}$, we get
\begin{equation*}
    \sum_{k=1}^n\frac{a_k^{(n)}}{x+\theta_k^{(n)}} = \sum_{\ell=1}^{n/2}
    2\mathcal{R}e\left(\frac{a_{2\ell}^{(n)}}{x+\theta_{2\ell}^{(n)}}\right).
\end{equation*}
It follows that the number of computations can be divided by two. The same holds for the computation of ${\mathcal R}_n(A)$ when $Spec(A)$ is real, e.g., in the Hermitian case that we consider in this paper. 
\end{remark}
\subsection{Setting}
Because of the results obtained in Section~\ref{Sec:Floating}, we consider cases where $n\leq 32$ so that floating arithmetic does not affect our results. All computational times are measured thanks to the MATLAB/Octave \texttt{tic} / \texttt{toc} functions.
The simulated parallel computational times for our method are estimated as follows: $t_{para} = \max_{1\leq i \leq n}{t_i}$ where $t_i$ is the computational time of the $i$-th matrix inversion or linear system resolution. We denote by $t_{seq}$ the sequential computing time of the approximation via \texttt{expm}
and by $t_{Krylov}$ 
the time taken by a rational Krylov approximation to get  the same absolute error as the one of our method using $n=32$.

We first show that the error of our method does not depend on the 
dimension of the matrix 
by considering a symmetric $d\times d$ real matrix $A$, with spectrum randomly chosen within a fixed range.
We then compare with the above mentioned algorithms, using matrices $B = \Delta^{1}_d\in{\cal M}_d({\mathbb R})$ and $C = \Delta^{2}_d\in{\cal M}_d({\mathbb R})$ corresponding respectively to the usual Finite Difference discretization of the Laplace operator in one and two dimensions. We consider both approximations $\exp(\cdot)$ and $\exp(\cdot)v$ where $v$ is a random vector of size $d$ and norm $1$.

\subsection{Stability of the error with respect to the dimension}
Given a matrix $A$, we consider either the absolute error $\|\exp(A) - \mathcal{R}_n(A)\|_2$ or the relative error $\frac{\|\exp(A) - \mathcal{R}_n(A)\|_2}{\|\exp(A)\|_2}$ when the spectrum of $A$ is non-positive or include positive eigenvalues, respectively. This choice follows from the fact that relative error is relevant for large numbers whereas small numbers are correctly analysed with absolute error. Both cases occur when using $\exp$.

As a first example, we consider the matrix $A$ described previously. 
The absolute error $\|\exp(A) - \mathcal{R}_n(A)\|_2$ as a function of the dimension $d$ is represented in Figure \ref{fig:matrix1} (left panel). 
The results are smoothed by using the mean of the error for various random spectra included in $[-1,0]$. We use the approximation~\eqref{equ:ExpoMatRat} where the inverse matrix is computed using the functions \texttt{inv} of MATLAB and Octave. We note that the error does not depend on $d$, which is an expected result since the spectrum remains in a fixed interval.
\begin{figure}[ht]
\begin{center}
\includegraphics[width=.46\textwidth]{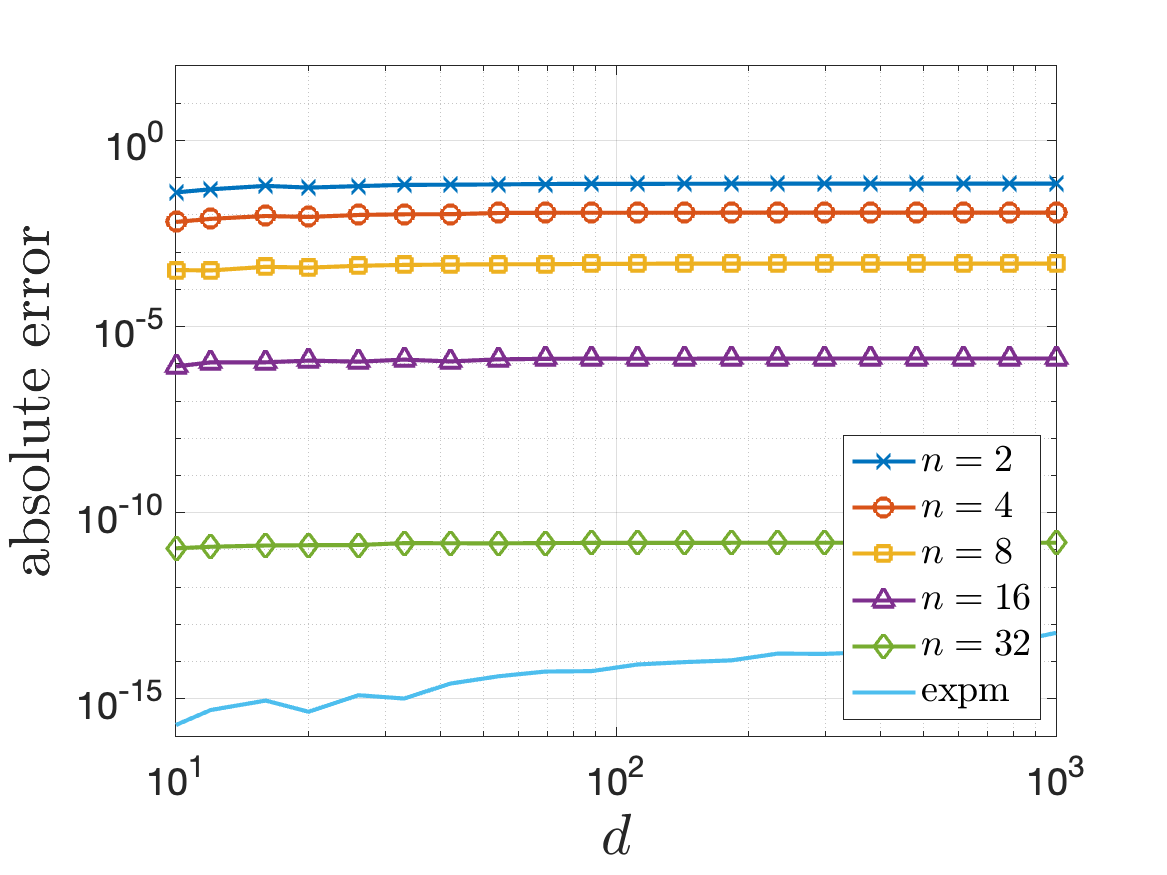} 
\includegraphics[width=.46\textwidth]{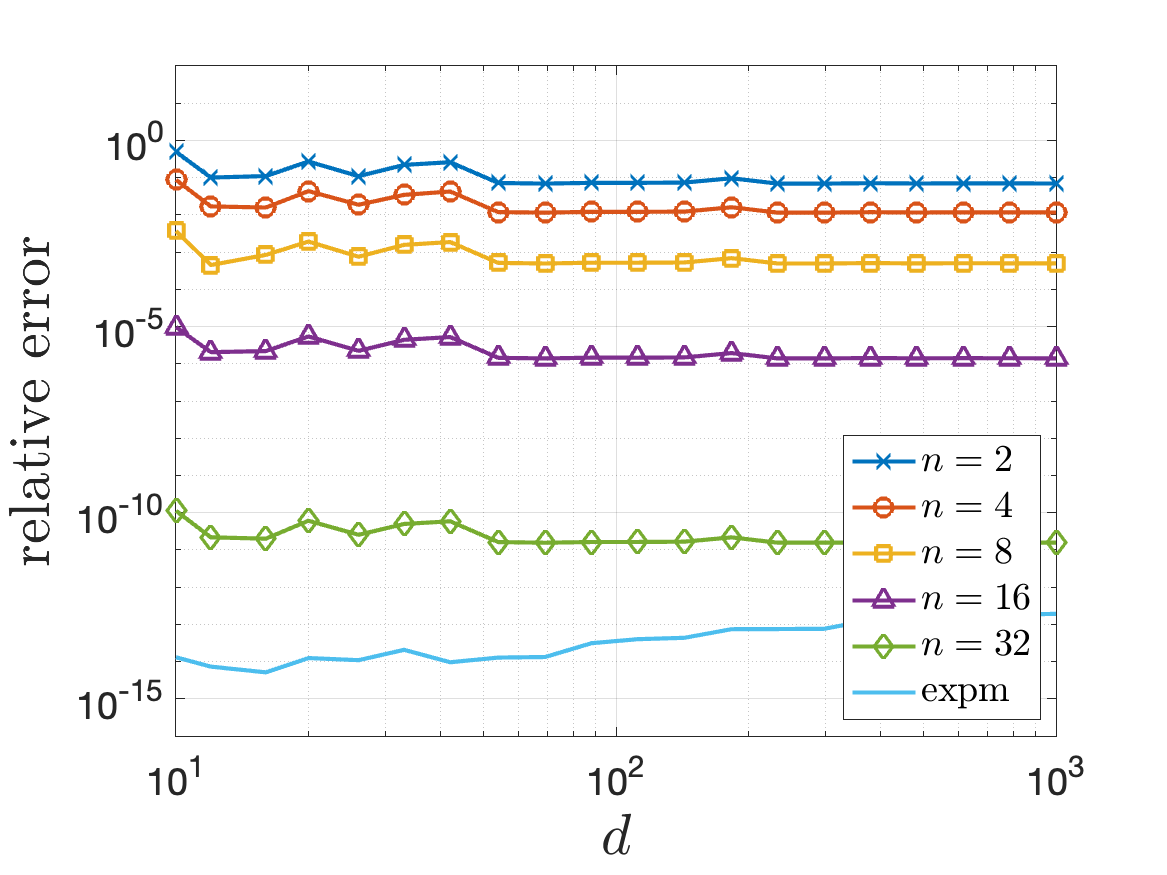} 
\end{center}
\caption{Error versus the  dimension of the matrix. Left: absolute error, matrices with negative spectra.  Right: relative error, matrices with positive spectra using the shift method from Remark \ref{rem:shift}.}
\label{fig:matrix1}
\end{figure}
In a second example, we consider a matrix with positive spectrum included in $[0,20]$. We use the shift method presented in Remark~\ref{rem:shift} to compute the exponential, and the relative error $\frac{\|\exp(A) - \mathcal{R}_n(A)\|_2}{\|\exp(A)\|_2}$. We see that the error still  not depends on the dimension of the problem. 
\subsection{Computation of $\exp(\cdot)$ for the Laplace operator}
From now on, we focus on matrices $B$ and $C$ described previously.

The absolute error is computed in practise by $\|\texttt{expm}(\cdot) - \mathcal{R}_n(\cdot)\|_2$. The results are presented in Figure \ref{fig:matrix3_1D}  (bottom) for $B$ and in Figure \ref{fig:matrix3_2D} for $C$ (bottom).
Here again, the error does not depend on the dimension of the problem, but only on the degree of truncation $n$. Note that increasing $n$ only expands the spectrum of these matrices on the left side, hence does not deteriorate our approximations.


Next, we compare the computing times $t_{seq}$ and $t_{para}$. The results are presented in Figure \ref{fig:matrix3_1D} for $B$ and in Figure \ref{fig:matrix3_2D} for $C$. 
In these tests, the matrices $(B+\theta_k^{(n)}I)^{-1}$ are computed with the MATLAB and Octave functions \texttt{inv} in parallel 
and, as explained previously $t_{para}$ is defined as the maximum time taken to compute one of the $a_k^{(n)} (A+\theta_k^{(n)}I)^{-1}$.
We can see that $t_{para}$ is slightly larger than $t_{seq}$ in the case of MATLAB 
and almost ten times larger in the case of Octave.

 \begin{figure}[ht]
\begin{center}
\includegraphics[width=6cm]{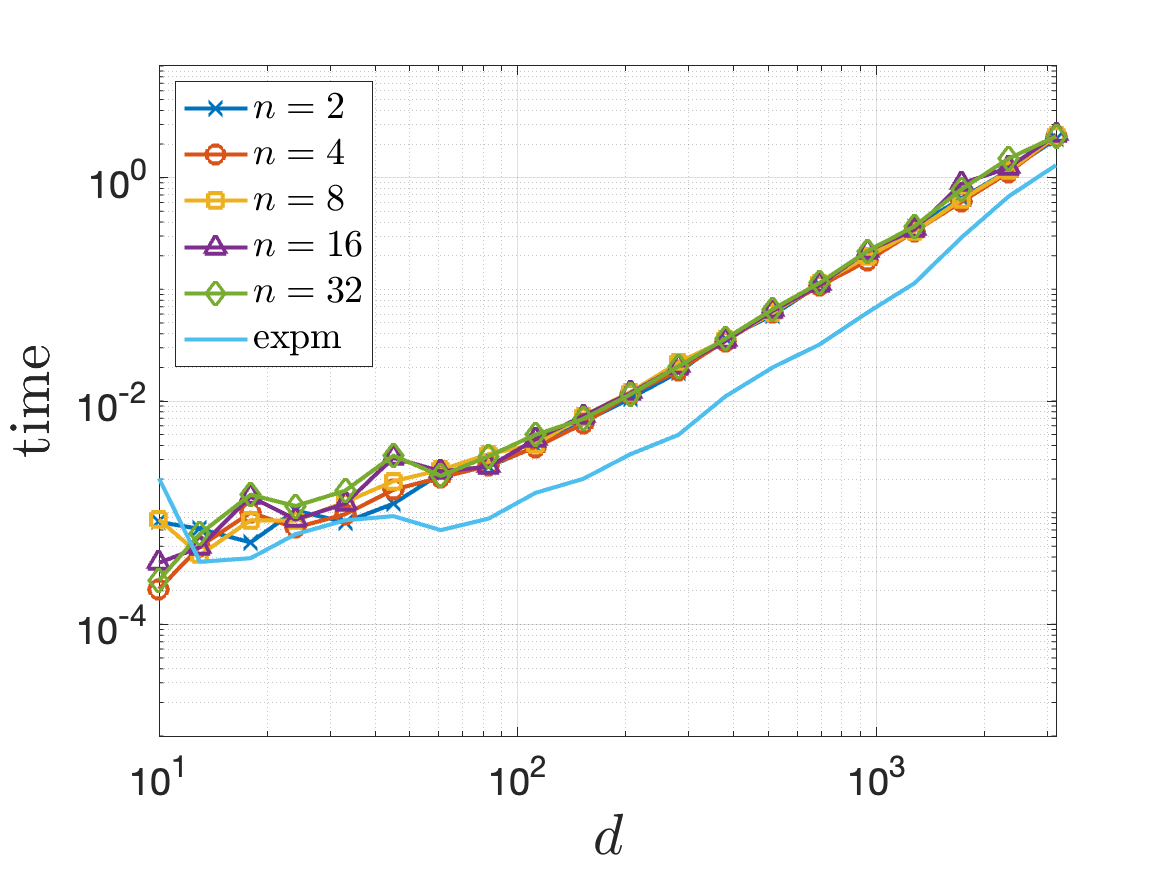}  
\includegraphics[width=6cm]{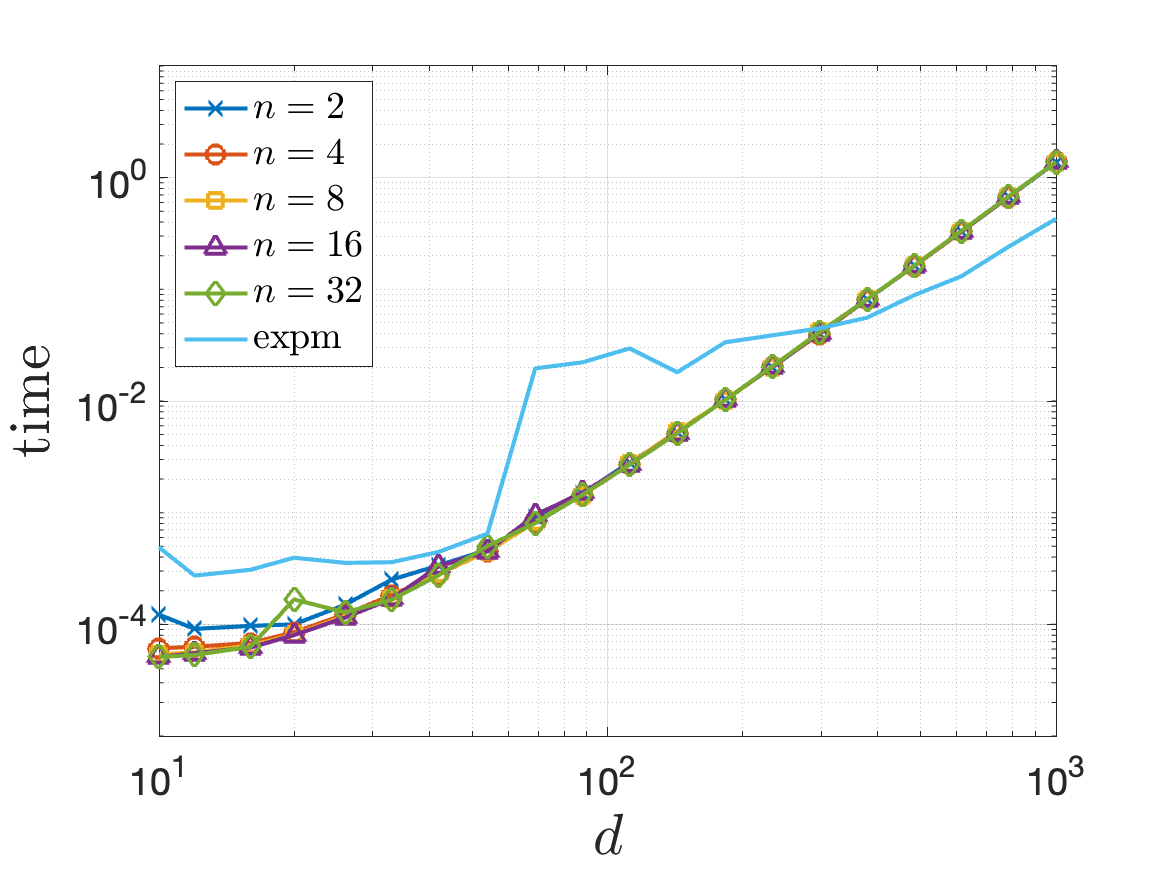}  
\includegraphics[width=6cm]{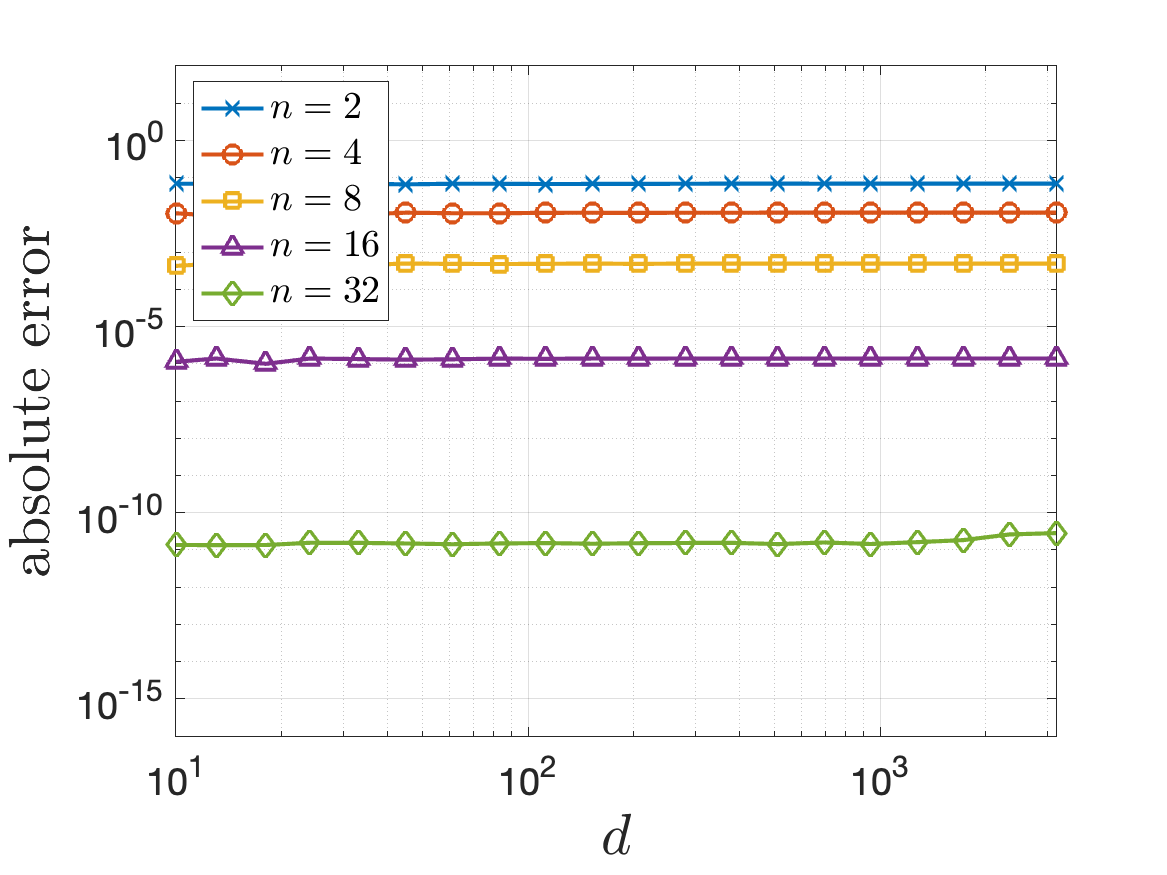}
\end{center}
\caption{Performance of the approximation of $\exp(B)$ for matrix $B = \Delta^{1}_d$.
Top:  CPU time required to compute \texttt{expm} and  approximation ${\cal R}_n$ for various values of $n$ using MATLAB (left) and  Octave (right).
Bottom: relative error in the computation of $\exp(B)$ as a function of $d$.}
\label{fig:matrix3_1D}
\end{figure}

 \begin{figure}[ht]
\begin{center}
\includegraphics[width=6cm]{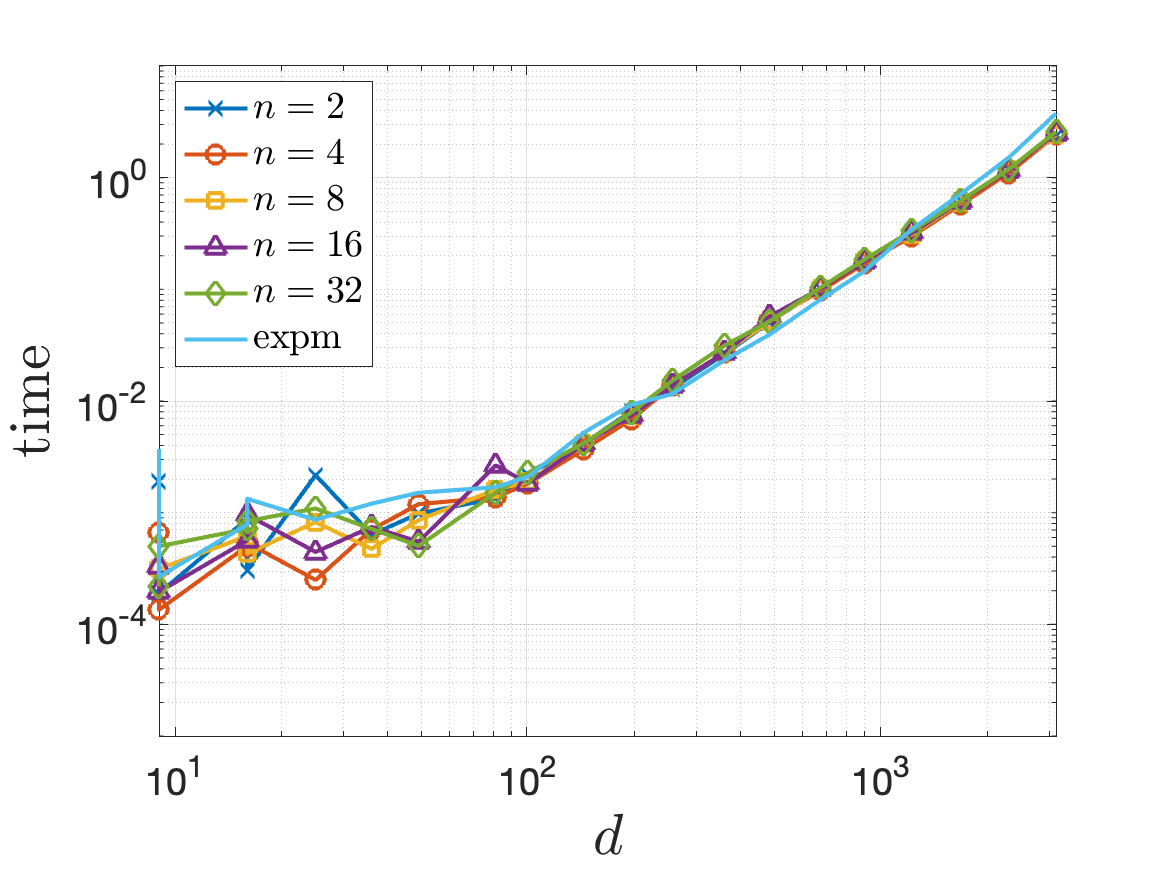}  
\includegraphics[width=6cm]{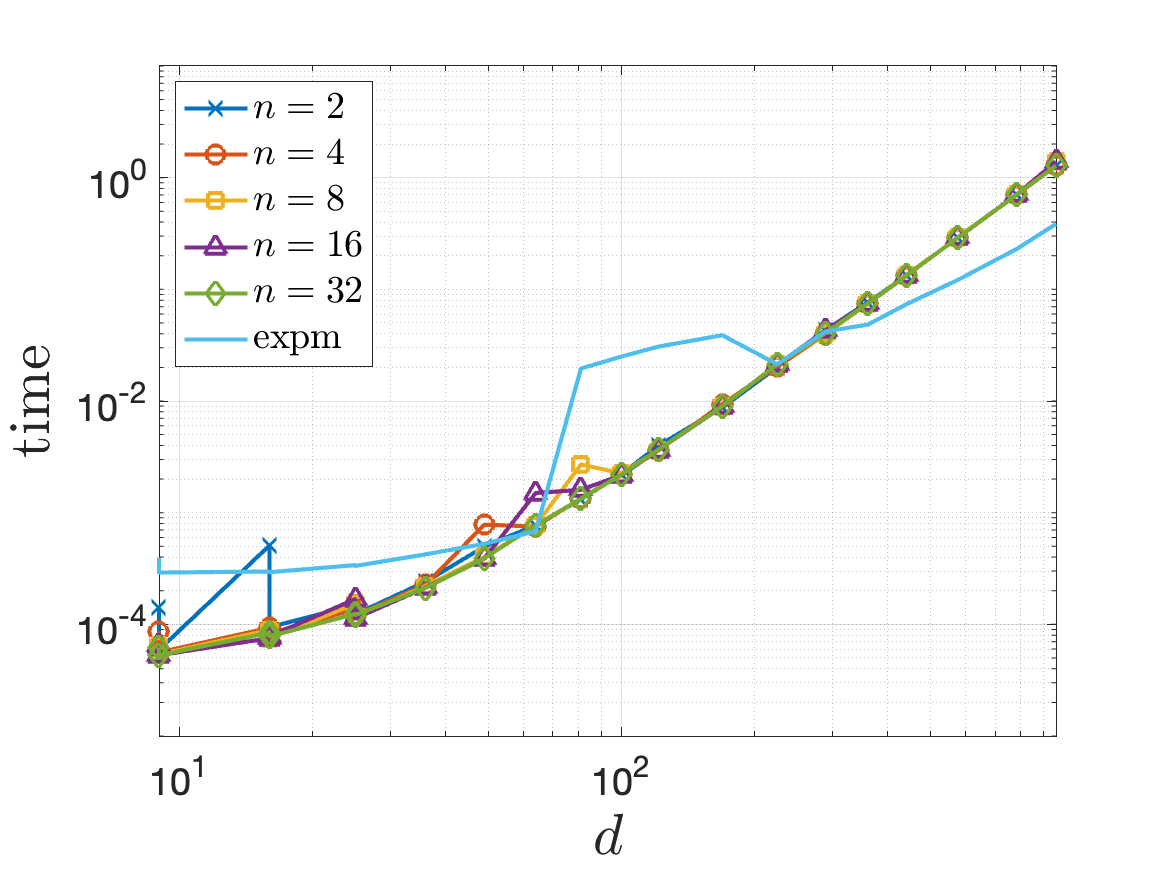}  
\includegraphics[width=6cm]{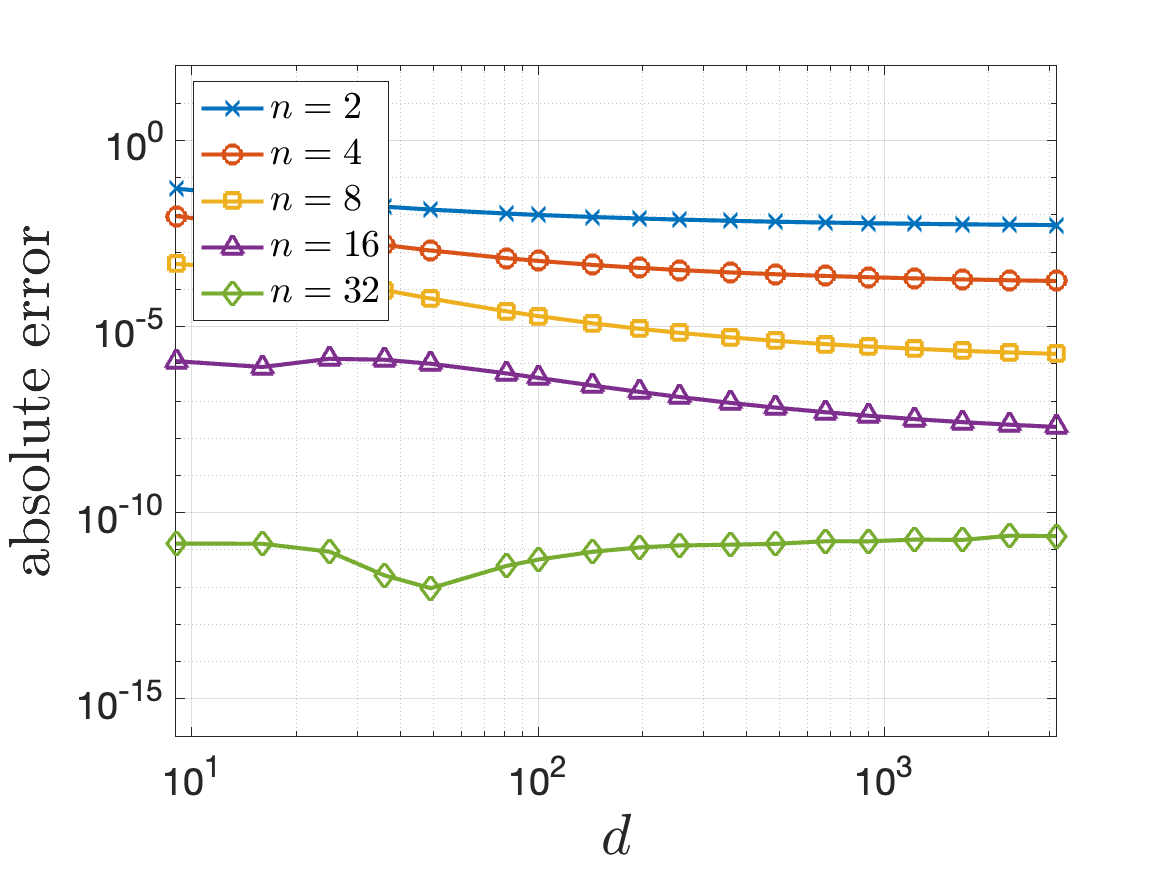}
\end{center}
\caption{Same as Figure~\ref{fig:matrix3_1D} for matrix $C = \Delta^{2}_d$.
}
\label{fig:matrix3_2D}
\end{figure}

\subsection{Computation of $\exp(\cdot)v$ for the Laplace operator}

Still considering matrices $B$ and $C$, we finally consider the action of the matrix exponential on vectors. For $v \in \mathbb{R}^d$, the vector $w = \exp(B) v$ is approximated by (\ref{equ:expAx}) where $(B+\theta_k^{(n)}I)y_k^{(n)}=v$ is solved using the solvers \texttt{mldivide} of MATLAB and Octave. We evaluate the mean of the error and the mean of $t_{para}$ for a series of random vectors $v$ of norm $1$.

Rational Krylov methods being the state of the art for this type of computation, we first look at the method proposed in \cite{guttel} to compare it with our approximation. All our  results related to rational Krylov methods are obtained using Guettel's toolbox \cite{guettel_2020_toolbox}, with parameter $\xi = -1$ as in \cite{guttel}.
Figure~\ref{fig:error_vs_time} shows the absolute error of the two methods, as a function of the computational time. 
We see that in both cases, for a prescribed accuracy, our method outperforms rational Krylov method.

\begin{figure}[ht]
\begin{center}
\includegraphics[width=.46\textwidth]{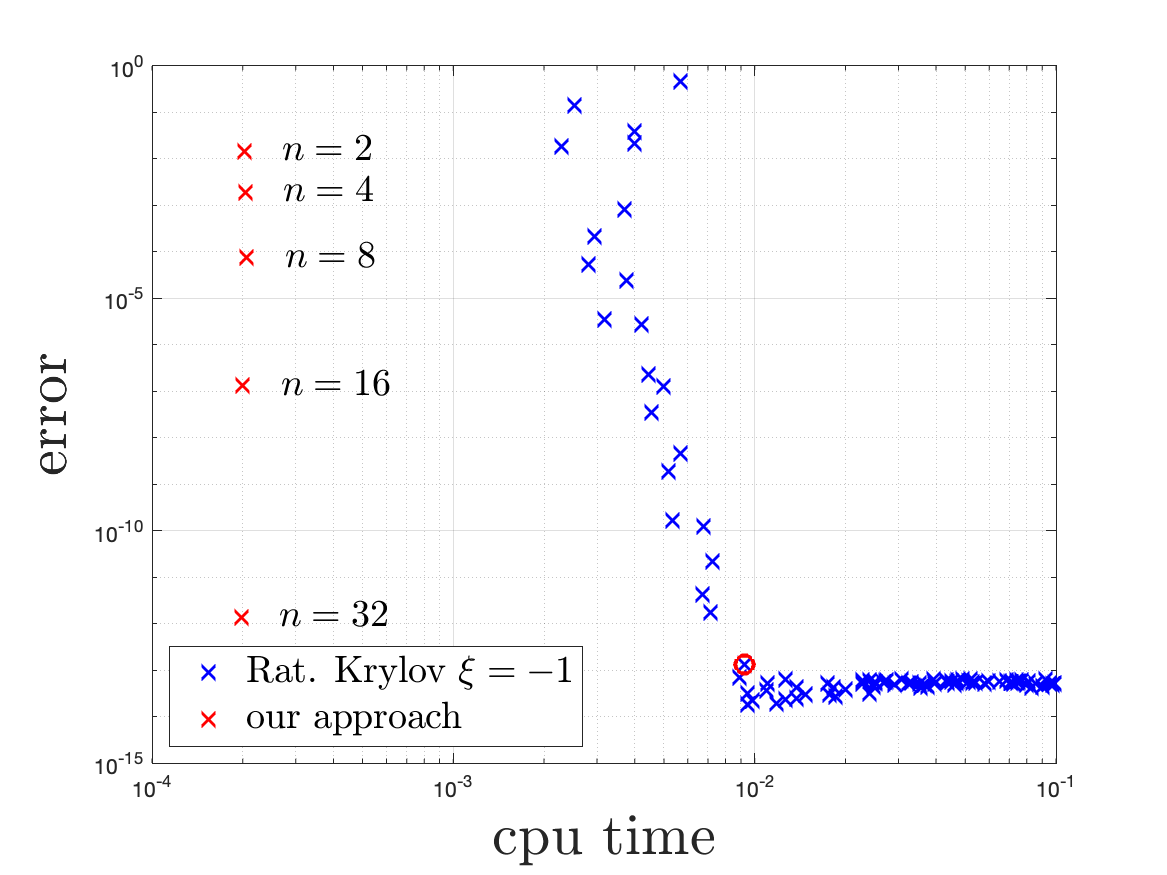} 
\includegraphics[width=.46\textwidth]{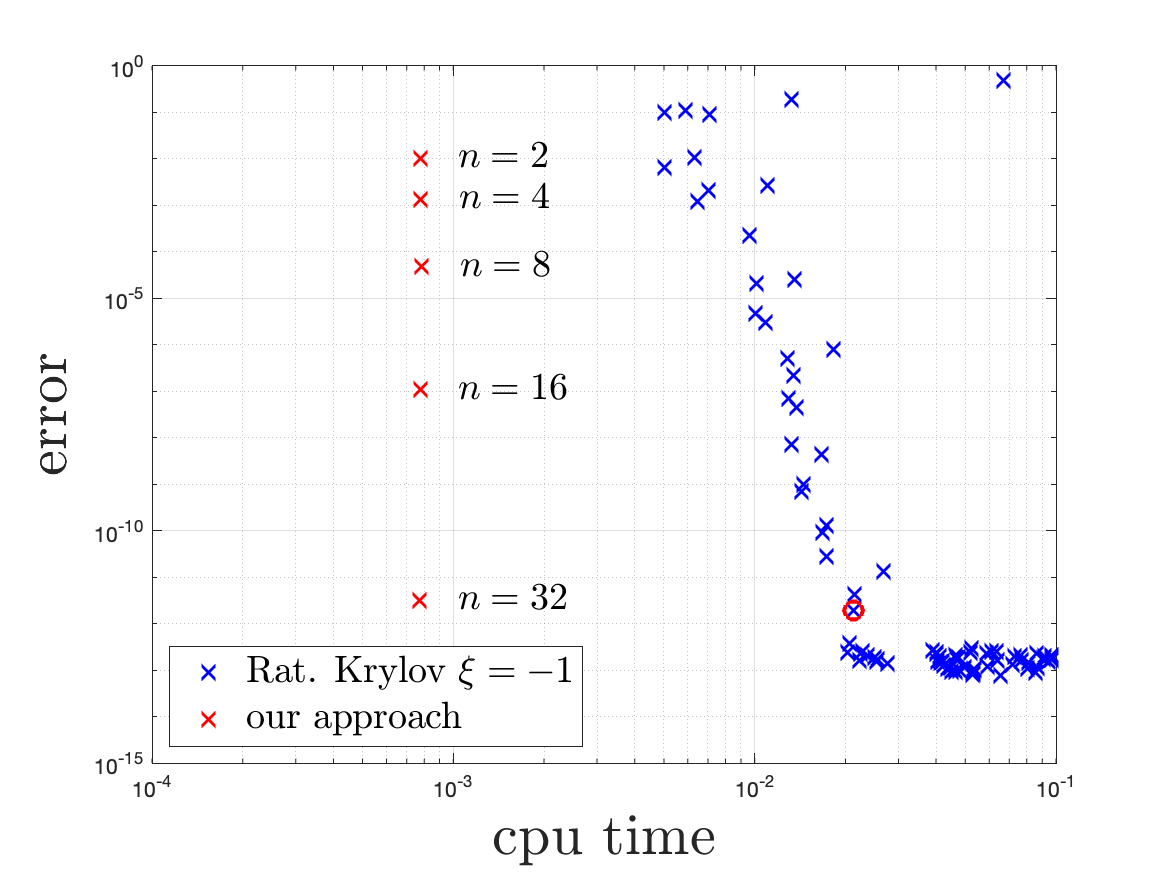} 
\end{center}
\caption{Error and CPU time comparison of the approximation $\exp(B)v$  using rational Krylov method (blue) and our approach (red). Vector $v$ is randomly chosen,  $d = 100$ (left) and $d = 400$ (right). Each point of the cloud corresponds to a Krylov rational space of size $m$ with $m \in \{1,...,d-1\}$. The dimension  $m$ for which Krylov method is at least as precise as our method with $n = 32$  is circled in red.}
\label{fig:error_vs_time}
\end{figure}

 The absolute error $\|\texttt{expm}(B)v - \mathcal{R}_n(B)v\|_2$  together with the computing times  $t_{seq}$, $t_{Krylov}$ and $t_{para}$ are shown in 
 Figure~\ref{fig:matrixexpAb_1D}.
 In this case $t_{para}$ is defined as the maximum time used to compute one of the vectors $a_k^{(n)} y_k^{(n)}$ (see~\eqref{equ:expAx}). 
We note that $t_{seq}$ and $t_{Krylov}$ are larger than $t_{para}$ for all values of the dimension $d$ of the matrix, with MATLAB as well as Octave. For example, with $d = 10^3$, $t_{seq} \approx 10^2 t_{para}$ and $t_{seq} \approx 10^4 t_{para}$ with MATLAB and Octave, respectively, and $t_{Krylov} \approx 10 t_{para}$ and $t_{Krylov} \approx 10^3 t_{para}$.

\begin{figure}[ht]
\begin{center}
\includegraphics[width=6cm]{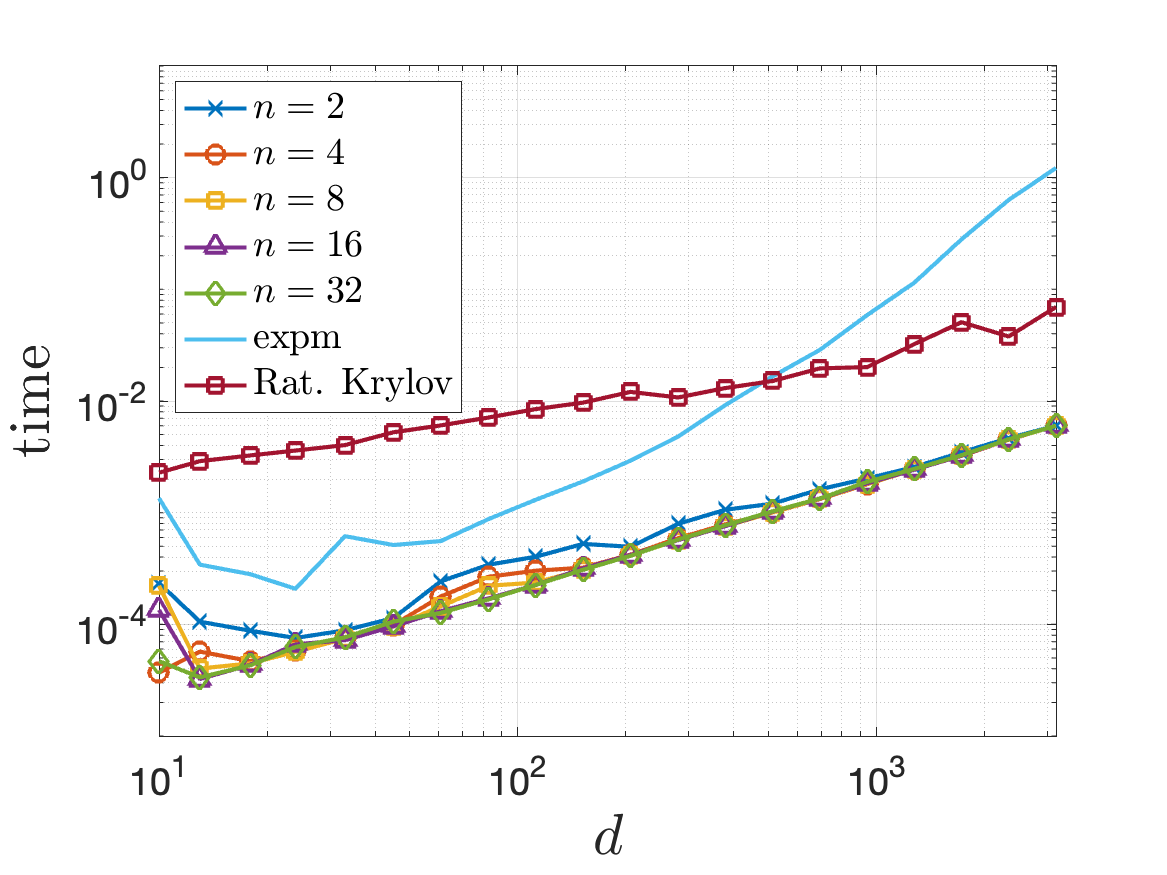}
\includegraphics[width=6cm]{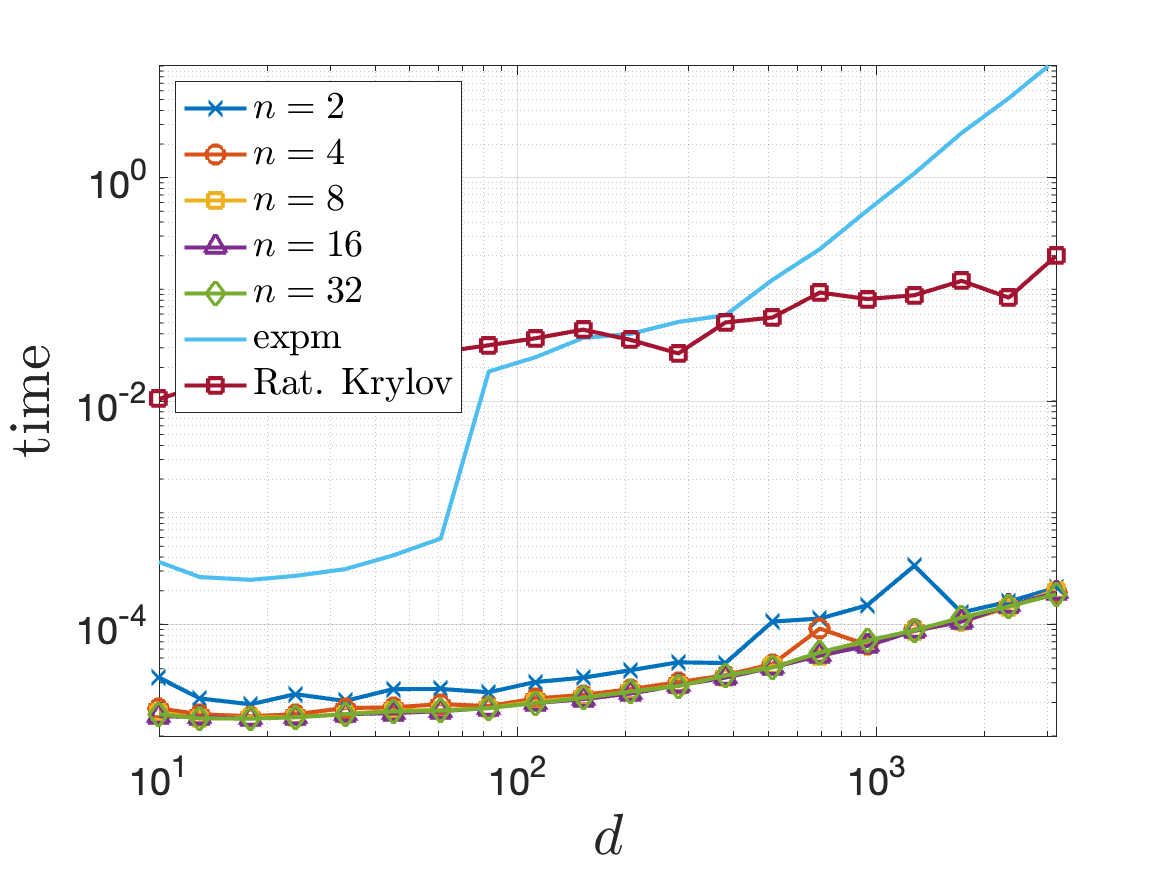}
\includegraphics[width=6cm]{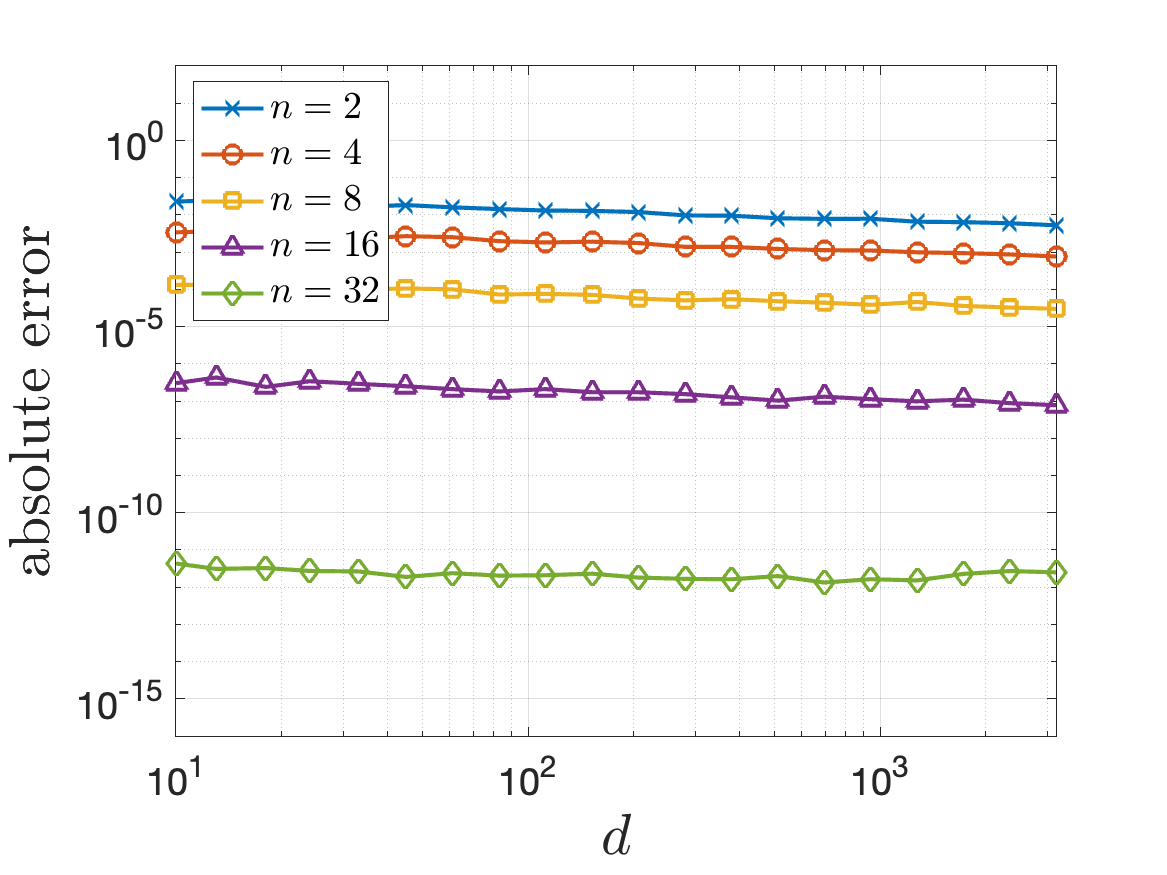}
\end{center}
\caption{Top: CPU time required to compute  $Av = \exp(\Delta^{1}_d)v$ for vectors $v$ randomly chosen using  rational Krylov method  and  our method for various values of $n$. Computations are done with  MATLAB
 (left) and Octave (right).  
Bottom: relative error, as a function of the dimension of  the matrix. 
}
\label{fig:matrixexpAb_1D}
\end{figure}

We make the same analysis for the absolute error $\|\texttt{expm}(C)v - \mathcal{R}_n(C)v\|_2$ together with the computing times  $t_{seq}$, $t_{Krylov}$ and $t_{para}$ are shown in 
 Figure~\ref{fig:matrixexpAb_2D}.
We note again that $t_{seq}$ and $t_{Krylov}$ are larger than $t_{para}$ for all values of the dimension $d$ of the matrix, with MATLAB as well as Octave. For $d = 10^3$, we observe that $t_{seq} \approx 10 t_{para}$ and $t_{seq} \approx 10^2 t_{para}$ with MATLAB and Octave, respectively, whereas $t_{para} \lessapprox t_{Krylov}$ 
and 
$t_{Krylov} \approx 10 t_{para}$.

\begin{figure}[ht]
\begin{center}
\includegraphics[width=6cm]{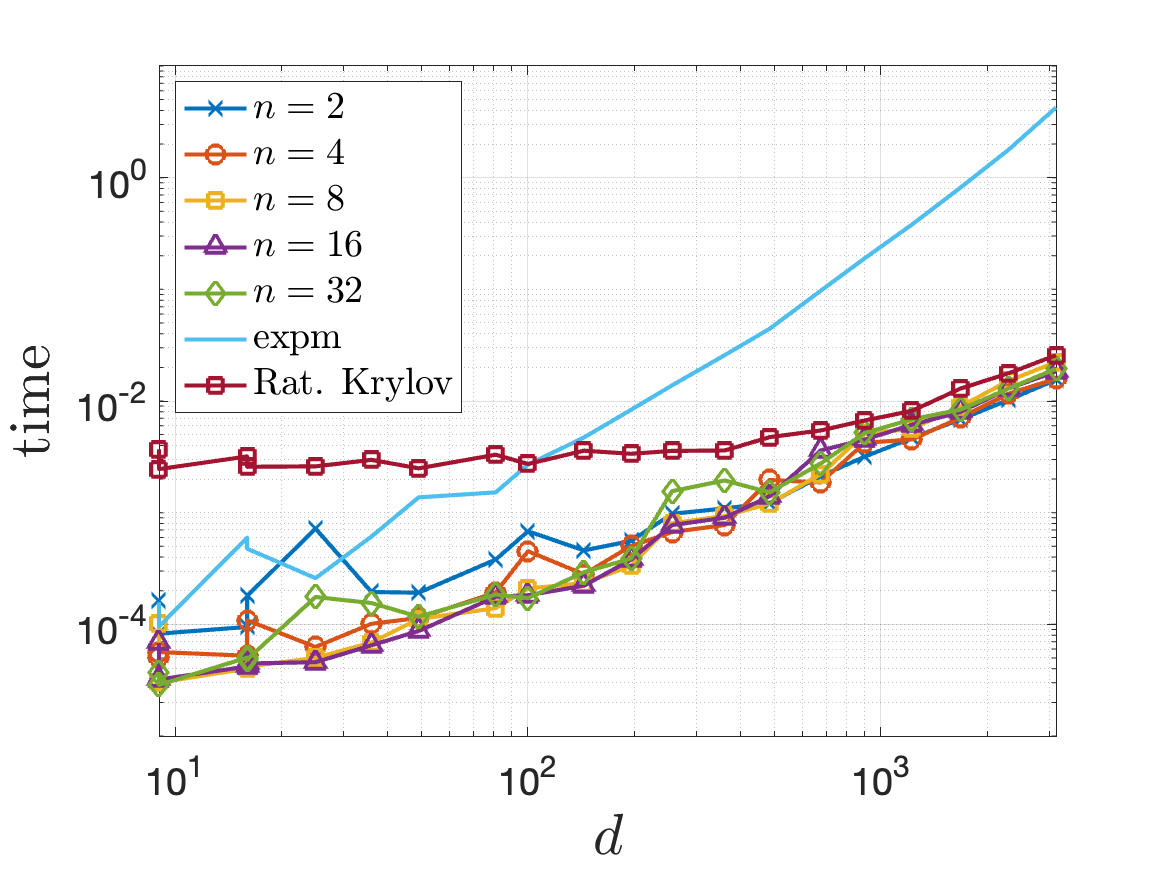}
\includegraphics[width=6cm]{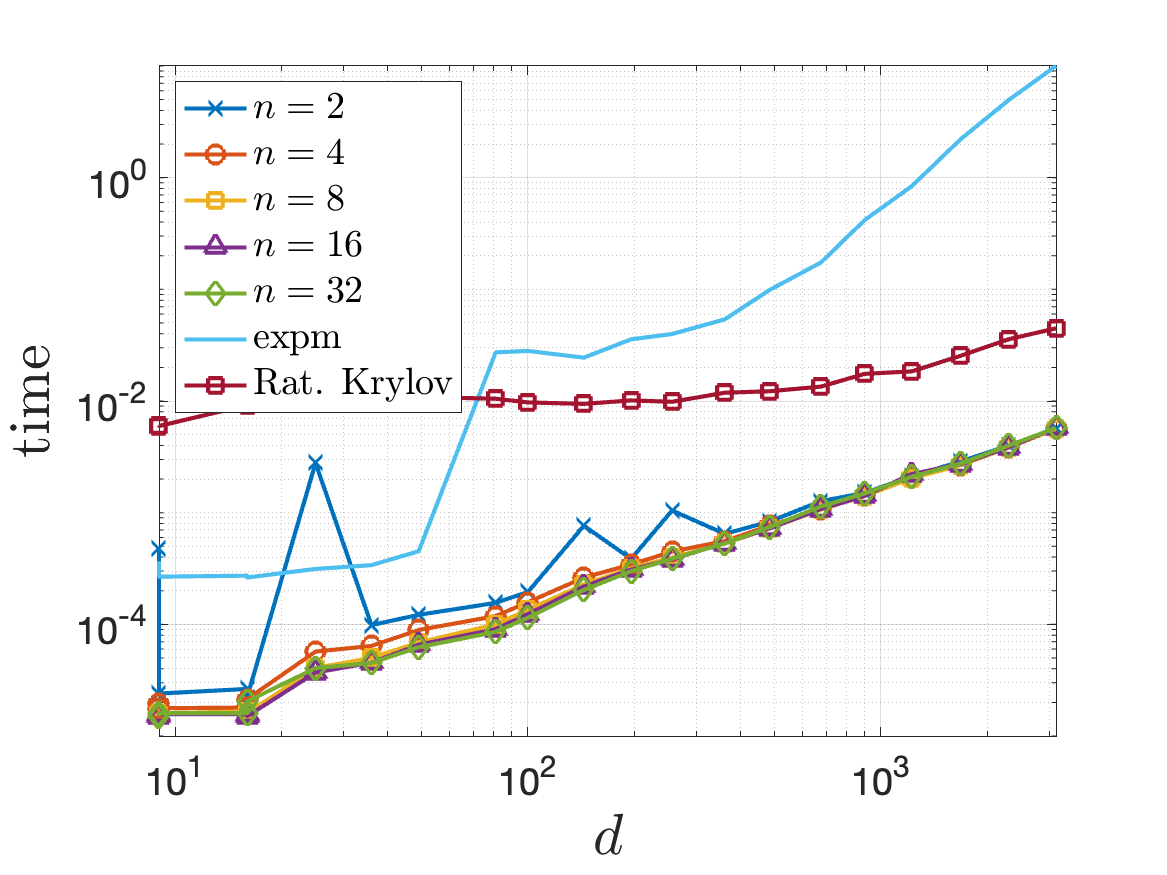}
\includegraphics[width=6cm]{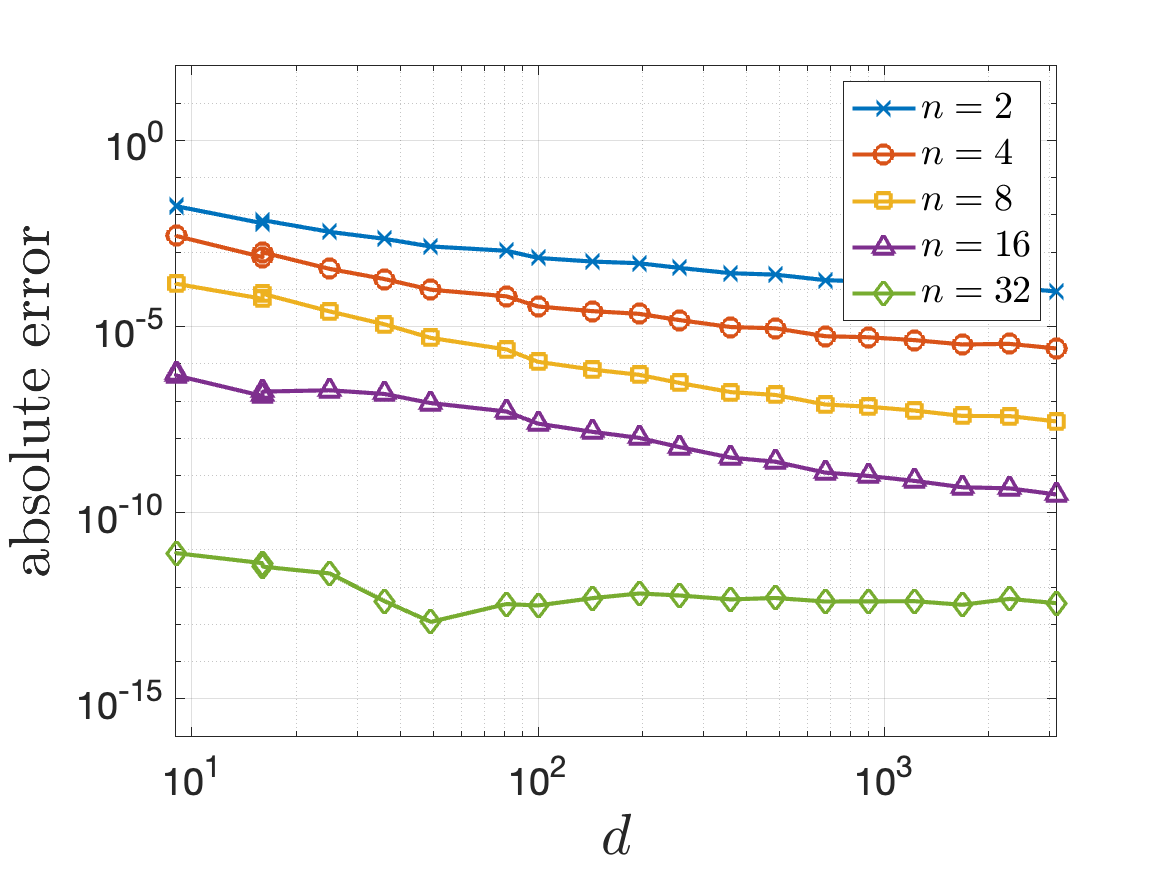}
\end{center}
\caption{Same as Figure~\ref{fig:matrixexpAb_1D} for matrix $B = \Delta^{2}_d$. 
}
\label{fig:matrixexpAb_2D}
\end{figure}


Finally,  we point out that our method defines a way to approximate the exponential of a matrix, whereas Krylov's rational methods approximate the matrix-vector product. 
These methods  are based on a  reduction of dimensionality. For the problems considered in this article they lead in practice to the computation of the exponential of a smaller matrix. 
Hence rational Krylov can be combined with our method for the computation of the exponential of this smaller matrix.



\bibliographystyle{abbrv}
\bibliography{HKPS}

\end{document}